\documentclass[12pt]{article}

\usepackage{amsmath,amsfonts,amsthm,amssymb,graphicx,epic,eepic,multicol}

\theoremstyle{definition}
\newtheorem{definition}{Definition}
\theoremstyle{remark}
\newtheorem{remark}[definition]{Remark}
\newtheorem{example}[definition]{Example}

\newtheoremstyle{mytheorem}{0.5cm}{0.2cm}{\slshape}{ }{\bfseries}{.}{ }{}
\theoremstyle{mytheorem}
\newtheorem{theorem}[definition]{Theorem}

\newtheorem{lemma}[definition]{Lemma}
\newtheorem{cor}[definition]{Corollary}

\newcommand{\E}{\mathbf{E}}
\newcommand{\EQ}{\mathbf{E}}
\newcommand{\R}{\mathbb{R}}
\newcommand{\EE}{\mathbb{E}}

\DeclareMathOperator{\conv}{conv}
\DeclareMathOperator{\one}{{1\hspace*{-0.55ex}I}}
\DeclareMathOperator{\grad}{grad}
\DeclareMathOperator{\SD}{SD}
\DeclareMathOperator{\ESD}{ESD}
\DeclareMathOperator{\QSD}{QSD}
\DeclareMathOperator{\BC}{BC}
\DeclareMathOperator{\GC}{GC}
\DeclareMathOperator{\BP}{BP}
\DeclareMathOperator{\GP}{GP}

\renewcommand{\P}{\mathbf{P}}

\newcommand{\Q}{\mathbf{Q}}

\newcommand{\sE}{\mathcal{E}}

\newcommand{\fF}{\mathfrak{F}}
\newcommand{\salg}{\fF}

\renewcommand{\phi}{\varphi}
\renewcommand{\kappa}{\varkappa}
\newcommand{\imagi}{\boldsymbol{\imath}}

\newcommand{\gammaQ}{\gamma}
\newcommand{\muQ}{\mu}
\newcommand{\nuQ}{\nu}

\newcommand{\thf}{\frac{1}{2}\,}
\newcommand{\tn}{|\!|\!|}
\newcommand{\tnc}{\tn\hspace{-1pt}\cdot\hspace{-1pt}\tn}
\newcommand{\pvint}{\int_{\R}\hspace{-14pt}-}

\newlength{\querylen}
\usepackage{fancybox}

\numberwithin{equation}{section}
\numberwithin{definition}{section}

\begin{document}
\bibliographystyle{plain}

\title{Geometric extension of put-call symmetry in the multiasset setting}

\author{Ilya Molchanov and Michael Schmutz\\
  \small Department of Mathematical Statistics and Actuarial Science,\\
  \small University of Bern, Sidlerstrasse 5, 3012 Bern, Switzerland\\
  \small (e-mails: ilya.molchanov@stat.unibe.ch, michael.schmutz@stat.unibe.ch)
}

\maketitle

\begin{abstract}
  In this paper we show how to relate European call and put options on
  multiple assets to certain convex bodies called lift zonoids. Based
  on this, geometric properties can be translated into economic
  statements and vice versa. For instance, the European call-put
  parity corresponds to the central symmetry property, while the
  concept of dual markets can be explained by reflection with respect
  to a plane.  It is known that the classical univariate log-normal
  model belongs to a large class of distributions with an extra
  property, analytically known as put-call symmetry. The geometric
  interpretation of this symmetry property motivates a natural
  multivariate extension. The financial meaning of this extension is
  explained, the asset price distributions that have this property are
  characterised and their further properties explored. It is also shown
  how to relate some multivariate asymmetric distributions to
  symmetric ones by a power transformation that is useful to adjust for
  carrying costs. A particular attention is devoted to the case of
  asset prices driven by L\'evy processes.  Based on this, semi-static
  hedging techniques for multiasset barrier options are suggested.

  \medskip


\end{abstract}

\newpage

\section{Options and zonoids: an introduction}
\label{sec:norms-stop-loss}

The stop-loss transformation from mathematical insurance theory
associates a random variable $\zeta$ with its stop-loss
$\E(\zeta-k)_+$ at level $k$. Here $\E$ denotes the expectation and
$x_+=\max(x,0)$ for a real number $x$, where $k$ is usually
interpreted as the excess (part of claim which is not paid by the
insurer).

From the mathematical finance viewpoint, the stop-loss
transformation is identified as the expected payoff from a European
call option for $\zeta=S_T=S_0e^{rT}\eta$ being the price of a (say
non-dividend paying) asset at the maturity time $T$, where $S_0$ is
the spot price and $e^{rT}\eta$ is the factor by which the price
changes, $r$ is the (constant) risk-free interest rate and $\eta$ is
an almost surely positive random variable.  In arbitrage-free and
complete markets the expectation can be taken with respect to the
unique equivalent martingale measure, so that the expected value of
the discounted payoff becomes the call price. If the underlying
probability measure is a martingale measure, then $\E\eta=1$ and the
discounted price process $S_te^{-rt}$, $t\in[0,T]$, becomes a
martingale. Unless indicated by a different subscript, all
expectations in this paper are understood with respect to the
probability measure $\Q$, which is not necessarily a martingale
measure. In this paper we do not address the choice of a martingale
measure in incomplete markets.

The expected call payoff $\E(S_T-k)_+$ can be considered a function of
the bivariate vector $(k,F)$, where
\begin{displaymath}
  F=S_0e^{rT}
\end{displaymath}
is the theoretical \emph{forward price} on the same asset with the
same maturity, i.e.\ $\E(S_T-k)_+=\E(F\eta-k)_+$. Deterministic
dividends or income until maturity can be incorporated in the forward
price, e.g.\ by setting $F=S_0e^{(r-q)T}$ in case of a continuous
dividend yield $q$.

When working with $n$ assets, we write $\eta$ for an $n$-dimensional
random vector $(\eta_1,\dots,\eta_n)$ such that the price $S_{Ti}$ of
the $i$th asset at time $T$ equals $F_i\eta_i$ with $F_i$ being the
corresponding forward price. We denote this shortly as
\begin{equation}
  \label{eq:det-eta}
  S_T=F\circ\eta=(F_1\eta_1,\dots,F_n\eta_n)\,.
\end{equation}

In order to relate the expected payoffs to certain convex sets we need
the following basic concept from convex geometry.

\begin{definition}[see~\cite{schn}, Sec.~1.7]
  \label{def:cs}
  The \emph{support function} of a nonempty convex compact set $K$ in
  the $n$-dimensional Euclidean space $\R^n$ is defined by
  \begin{displaymath}
    h_K(u)=\sup\{\langle x,u\rangle:\; x\in K\}\,,\quad u\in\R^n\,,
  \end{displaymath}
  where $\langle x,u\rangle$ is the scalar product in $\R^n$.
\end{definition}

For instance, if $K=[-x,x]$ is the line segment in $\R^n$ with
end-points $\pm x$, then $h_K(u)=|\langle x,u\rangle|$; if
$K=[0,x]$, then $h_K(u)=\langle x,u\rangle_+$; if $K$ is the
triangle in $\R^2$ with vertices $(0,0)$, $(a,0)$, and $(0,b)$,
then $h_K(u)=\max(u_1a,u_2b,0)$ for all $u=(u_1,u_2)\in\R^2$.

A function $g:\R^n\mapsto\R$ is called sublinear if it is positively
homogeneous ($g(cx)=cg(x)$ for all $c\geq0$ and $x\in\R^n$) and
subadditive ($g(x+y)\leq g(x)+g(y)$ for all $x,y\in\R^n$). It is well
known in convex geometry that support functions are characterised by
their sublinearity property and that there is a one-to-one
correspondence between support functions and \emph{convex bodies},
i.e.\ nonempty compact convex subsets of $\R^n$, see e.g.\
\cite[Th.~1.7.1]{schn}.

With each integrable $n$-dimensional random vector
$\eta=(\eta_1,\dots,\eta_n)$ it is possible to associate a
$(n+1)$-dimensional convex body which uniquely describes the
distribution of $\eta$. For this, consider $(n+1)$-dimensional random
vector $(1,\eta)$ obtained by concatenating $1$ and $\eta$ or, in
other words, by \emph{lifting} $\eta$ with an extra coordinate being
one.  In the financial setting this extra coordinate represents a
riskless bond.  Because of the lifting, we number the coordinates of
$(n+1)$-dimensional vectors as $0,1,\dots,n$ and write these vectors
as $(u_0,u)$ for $u_0\in\R$ and $u\in\R^n$ or as
$(u_0,u_1,\dots,u_n)$.

Let $X$ be the random set being the line segment in $\R^{n+1}$ with
end-points at the origin and $(1,\eta)$, see~\cite{mo1} for detail
presentation of random sets theory.  The support function of $X$ is
given by
\begin{displaymath}
  h_X(u_0,u)=\max(u_0+u_1\eta_1+\cdots+u_n\eta_n,0)=\langle (u_0,u),(1,\eta)\rangle_+
\end{displaymath}
for $(u_0,u)\in\R^{n+1}$. The integrability of $\eta$ implies that
$h_X(u_0,u)$ is integrable. The expected support function $\E h_X$ is
sublinear and so is the support function of a convex body $\E X$
called the (Aumann) \emph{expectation} of $X$, see
\cite[Sec.~2.1]{mo1}. For our choice of $X$, the set $\E X$ is called
the \emph{lift zonoid} of $\eta$ and denoted by $Z_\eta$. It is known
that $Z_\eta$ determines uniquely the distribution of $\eta$,
see~\cite[Th.~2.21]{mos02}. Note that the \emph{zonoid} of $\eta$
appears from a similar (non-lifted) construction as the expectation of
the random segment that joins the origin and $\eta$, see
\cite[Th.~2.8]{mos02}.

In the univariate setting we assume that $n=1$ and $\eta=S_T/F$ is a
positive random variable.  Then
\begin{equation}
  \label{eq:support-f-of-Q}
  h_{Z_\eta}(u_0,u_1)=\E(u_0+u_1\eta)_+=
  \begin{cases}
    u_0+u_1\E\eta      & u_0,u_1\geq 0\,,\\
    0                   & u_0,u_1< 0 \,,\\
    \E(u_0+u_1\eta)_+ &\text{otherwise }
  \end{cases}
\end{equation}
for all $(u_0,u_1)\in\R^2$.  Figure~\ref{fi:liftzonoids} shows the
lift zonoid of $\eta$ for various volatilities ($0.25,0.5,0.75$) in
the log-normal case with $\E\eta=1$ calculated for $T=1$. The higher
the volatility the larger (thicker) becomes the lift zonoid.  The
upper and lower boundaries of lift zonoids are the so-called
generalised Lorenz curves, which can be easily parametrised, see
\cite[pp.~43~and~44]{mos02}.

\begin{figure}[htbp]
  \begin{center}
    \includegraphics[height=6.5cm]{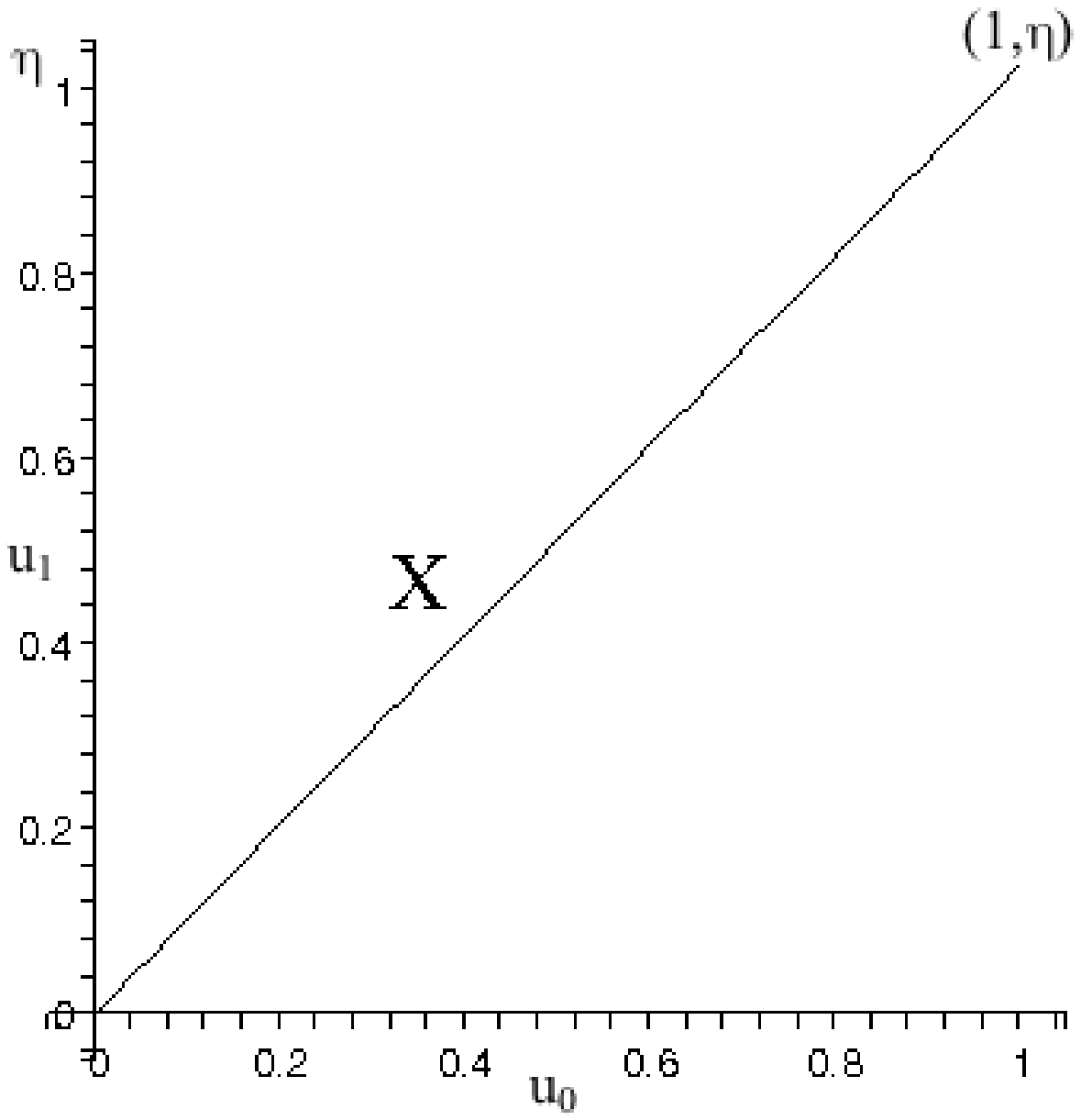}
    \includegraphics[height=6.5cm]{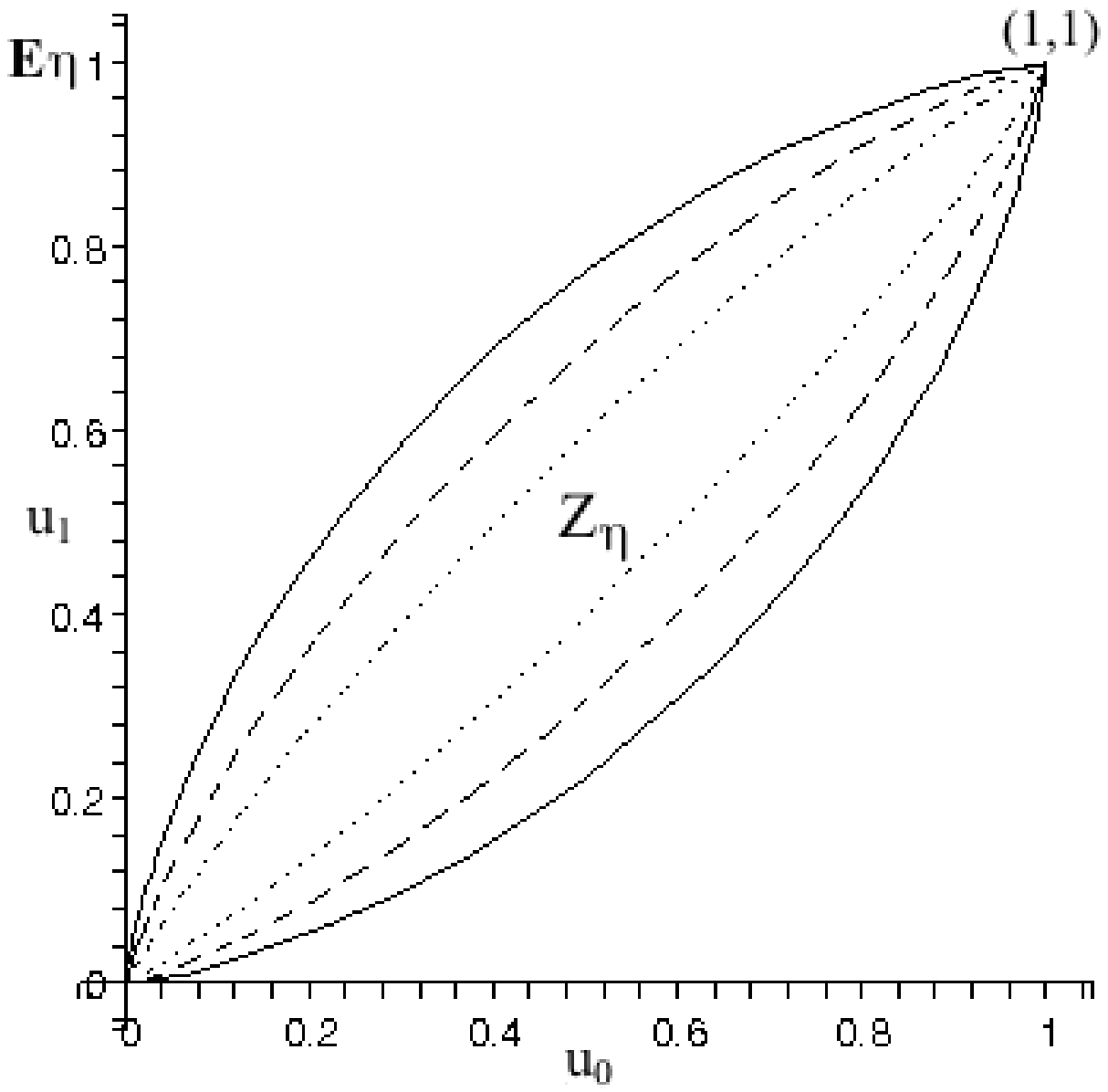}
  \end{center}
  \caption{The segment $X=[(0,0),(1,\eta)]$ and the lift zonoid
    $Z_\eta$ for log-normal $\eta$ with volatilities
    $\sigma=0.25,0.5,0.75$, drift $\mu=-\thf\sigma^2$ and maturity
    $T=1$.}
  \label{fi:liftzonoids}
\end{figure}

Since the lift zonoid uniquely determines the distribution of random
vector $\eta=(\eta_1,\dots,\eta_n)$, it also determines prices of
\emph{all} payoffs associated with $S_T=(S_{T1},\dots,S_{Tn})$,
assuming that the underlying probability measure is a martingale
measure.  For instance, in the univariate case $h_{Z_\eta}(-k,F)$
(resp.\ $h_{Z_\eta}(k,-F)$) is the non-discounted price of a
European call (resp.\ put) option with strike $k$. A similar
interpretation holds for basket options. The support function
determines uniquely the lift zonoid, so that the prices of European
vanilla options (basket calls and puts) determine uniquely the
distribution of the assets and so prices of all other European
options. In the univariate case this fact was noticed by Ross
\cite{ros76}, Breeden and Litzenberger~\cite{bre:lit78}, while Carr
and Madan \cite{car:mad94} presented an explicit decomposition of
general smooth payoff functions as integrals of vanilla options and
riskless bonds. In view of the positive homogeneity of support
functions and central symmetry of lift zonoids,
see~\cite[Prop.~2.15]{mos02}, it suffices in fact to have all call
prices with parameter vectors $(-k,u)$ with norm one in $\R^{n+1}$,
where $k>0$ and $u$ stands for the vector containing the products of
the weights and forward prices of the corresponding assets in the
components. Alternatively, it suffices to have all call prices for
any fixed $k>0$ and any $u\in\R^n$.

As just mentioned it is well known that the lift zonoid is
\emph{centrally symmetric}. Section~\ref{sec:symm-prop-lift} begins
by showing that the central symmetry property of lift zonoids is a
geometric interpretation of the call-put parity for European
options.

The main question addressed in this paper concerns further symmetry
properties of lift zonoids, their probabilistic characterisation and
financial implications. Section~\ref{sec:symm-prop-lift} continues
to show that in the univariate case (where lift zonoids are planar
sets) the reflection at the line bisecting the first quadrant
corresponds to the dual market transition at maturity. In view of
this, random variables that lead to line symmetric lift zonoids are
called \emph{self-dual}. This property has an immediate financial
interpretation as Bates' rule \cite{bat97} or the put-call symmetry
\cite{car94,car:ell:gup98}. For instance, the lift zonoids of the
log-normal distribution in the risk-neutral setting (see
Figure~\ref{fi:liftzonoids}) are line symmetric, which implies the
put-call symmetry (or Bates' rule) for the Black-Scholes economy.
Section~\ref{sec:symm-prop-lift} then shows how to translate the
geometric symmetry property into symmetry relationships for general
integrable payoffs and, in particular, for various binary and gap
options.

Section~\ref{sec:option-prices-norms} highlights relationships
between vanilla options and options on the maximum of the asset
price and a strike.  This leads to a concept of lift
\emph{max}-zonoids, which are particularly useful to describe
options involving maxima of possibly weighted assets. This section
also deals with a symmetry property of lift max-zonoids and shows
how to relate option prices to certain norms on $\R^2$ yielding a
relationship to the extreme values theory.

Section~\ref{sec:mult-symm} characterises random vectors that
possess symmetry properties generalising the classical put-call
symmetry for basket options and options on the maximum of several
assets.  The symmetry (or self-duality) is understood with respect
to each particular asset or for all assets simultaneously.
Relationships between the self-duality property and the
swap-invariance in Margrabe type options have been studied
in~\cite{mol:sch09}. In currency markets, the self-duality results
can be interpreted with respect to real existing markets, yielding
the basis for further applications, see~\cite{sch08}. The overall
symmetry implies that expected payoffs from basket options are
symmetric with respect to the weights of particular assets and the
strike price. We show that symmetries for some vanilla type options
(like Bates' rule in the univariate case) imply a certain symmetry
for \emph{every} integrable payoff function. After discussing some
fundamental results, we characterise the multivariate log-infinitely
divisible distributions, exhibiting the multivariate put-call
symmetry. The new effect in the multivariate setting is that
independence of asset prices prevents them from being jointly
self-dual. In other words, symmetry properties for several assets
enforce certain dependency structure between them, which is explored
in this paper.

In order to extend the application range of the self-duality
property and also in view of incorporating the carrying costs, we
then define \emph{quasi-self-dual} random vectors and characterise
their distributions. These random vectors become self-dual if their
components are normalised by constants (representing carrying costs)
and raised to a certain power. The related power transformation was
used in \cite[Sec.~6.2]{car:lee08} in the one-dimensional case. Here
we establish an explicit relationship between carrying costs and the
required power of transformation for rather general price models
based on L\'evy processes.

These results are then used in Section~\ref{sec:symm-distr} to
obtain several new results for self-dual random variables thereby
complementing the results from \cite{car:lee08}. In particular,
self-dual random variables have been characterised in terms of their
distribution functions; it is shown that self-dual random variables
always have non-negative skewness and several examples of self-dual
random variables are given.

As in the univariate case also in the multiasset case there are
various applications of symmetry results. First, symmetry results
may be used for validating models or analysing market data, e.g.\
similarly as in~\cite{bat97} and~\cite{faj:mor06} in the univariate
case. Furthermore, they could be used for deriving certain
investment strategies, see e.g.\
Section~\ref{sec:semi-static-super}. The probably most important
application will potentially be found in the area of hedging,
especially in developing \emph{semi-static} replicating strategies
of multiasset barrier and possibly also more complicated
path-dependent contracts. Following Carr and Lee~\cite{car:lee08},
semi-static hedging is the replication of contracts by trading
European-style claims at no more than two times after inception. As
far as the relevance of this application is concerned we should
mention that there has been a liquid market in structured products,
particularly in Europe. At the moment the majority of the trades is
still over-the-counter, but more and more trades are also organised
at exchanges, especially at the quite new European exchange for
structured products Scoach. Structured products often involve equity
indices, sometimes several purpose-built shares, and quite often
have barriers. Hence, developing robust hedging strategies for
multi-asset path-dependent products seems to be of a certain
importance. In the univariate case, Carr et
al.~\cite{car:cho97,car:cho02,car:ell:gup98,car:lee08} and several
other authors (see e.g.~\cite{and01,and:and:eli02,pou06}) developed
a machinery for replicating barrier contracts having fundamental
relevance for other path-dependent contracts.

Section~\ref{sec:ex-apl-self-dual} contains first applications of
the multivariate symmetry properties, especially for hedging complex
barrier options, thereby extending results from
\cite{car:cho97,car:cho02,car:ell:gup98} and \cite{car:lee08} for
some multiasset options. The development of a more general
multivariate semi-static hedging machinery is left for future
research.

\section{Symmetries of lift zonoids and financial relations for a
  single asset case}
\label{sec:symm-prop-lift}

\subsection{Parities}
\label{sec:parities}

We write $c(k,F)$ for the price of the European call option with
strike $k$ on the asset with forward price $F$.  Furthermore, let
$p(k,F)$ denote the price of the equally specified put.  The maturity
time $T$ is supposed to be the same for all instruments.

One of the most basic relationships between options in arbitrage-free
markets is the European \emph{call-put parity}.  In case of
deterministic dividends, this parity can be expressed by
\begin{equation}
  \label{eq:c-p-parity}
  c(k,F)=e^{-rT}(F-k)+p(k,F)\,.
\end{equation}

Recall that $\eta$ is defined by $S_T=F\eta$, where $S_T$ is the asset
price at maturity and $F$ is the forward price. The lift zonoid
$Z_\eta$ of $\eta$ is centrally symmetric around $\thf(1,\E\eta)$,
see~\cite[Prop.~2.15]{mos02}. If the expectation is taken with respect
to a martingale measure, then $\E\eta=1$, whence
$Z_{\eta,o}=Z_\eta-(\thf,\thf)$ is origin symmetric, so that
$h_{Z_{\eta,o}}(u)=h_{Z_{\eta,o}}(-u)$ for all $u\in\R^2$.
Interpreting the values of the support function of $Z_\eta$ as
non-discounted call and put prices, this symmetry yields that
\begin{align*}
  e^{rT}c(k,F)&=h_{Z_\eta}(-k,F)=h_{Z_{\eta,o}+(\thf,\thf)}(-k,F)\\
  &= h_{Z_{\eta,o}}(k,-F)-\thf k+\thf F\\
  &= h_{Z_{\eta,o}}(k,-F)+\thf k-\thf F-k+F\\
  &=h_{Z_\eta}(k,-F)-k+F=e^{rT}p(k,F)+F-k\,,
\end{align*}
i.e.\ we arrive at the classical European call-put parity. By defining
appropriate multidimensional lift zonoids and using their point
symmetry, the above proof can easily be extended to the call-put
parity for Asian options with arithmetic mean and to the European
call-put parity for basket options.

\begin{figure}[htbp]
  \begin{center}
    \includegraphics[height=6.5cm]{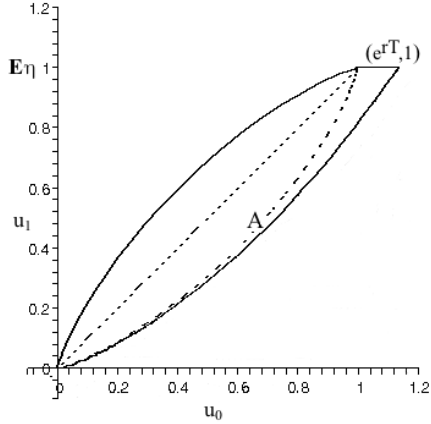}
  \end{center}
  \caption{An approximation of the payoff set $A$ for the
    Black-Scholes economy with volatility $\sigma=0.5$, interest rate
    $r=0.12$, dividend yield $q=0$ and maturity $T=1$.}
  \label{fi:am-subdifferential}
\end{figure}

It is also easy to explain geometrically why the parities do not
hold for American options.  Let $C(k,F)$ (resp.\ $P(k,F)$) be the
price of an American call (resp.\ put) with strike $k$ on the asset
with forward price $F$. Since the functions $C$ and $P$ are
sublinear, it is possible to define convex body $A$ with support
function $h_A(-k,F)=e^{rT}C(k,F)$ and $h_A(k,-F)=e^{rT}P(k,F)$, for
$k,F>0$. This convex body (that we call a \emph{payoff set})
determines the values of American vanilla options and is, up to some
rare exceptions, not centrally symmetric, see
Figure~\ref{fi:am-subdifferential}.

\subsection{Duality}
\label{sec:duality}

Recall that $\eta$ defined from $S_T=F\eta$ is almost surely
positive. If $\eta$ is distributed according to a martingale measure
$\Q$, then a new probability measure $\tilde\Q$ can be defined from
\begin{displaymath}
  \frac{d\tilde\Q}{d\Q}=\eta\,.
\end{displaymath}
Since $\eta$ is usually represented as $e^{H_T}$ for a semimartingale
$H_t$, $0\leq t\leq T$, and $-H_t$ is called the dual to $H_t$, the
random variable $\tilde\eta=\eta^{-1}$ is said to be the \emph{dual}
of $\eta$. The dual lift zonoid $Z_{\tilde\eta}$ is defined as the
$\tilde\Q$-expectation of the segment that joins the origin and
$(1,\tilde\eta)$.

\begin{lemma}
  \label{le:Z-of-dual}
  If $Z_\eta$ is the lift zonoid generated by almost surely positive
  random variable $\eta$ with $\E\eta=1$, then
  \begin{displaymath}
    Z_{\tilde\eta}=\tilde Z_\eta\,,
  \end{displaymath}
  where $\tilde Z_\eta$ denotes the reflection of $Z_\eta$ with respect to
  the line $\{(u_0,u_1)\in\R^2:\; u_0=u_1\}$.
\end{lemma}
\begin{proof}
  For $(u_0,u_1)\in\R^2$
  \begin{align*}
    h_{Z_{\tilde\eta}}(u_0,u_1)&=\E_{\tilde\Q}(u_0+u_1\tilde\eta)_+
    =\EQ\big[(u_0+u_1\eta^{-1})_+\eta\big]\\
    &=h_{Z_\eta}(u_1,u_0)=h_{\tilde Z_\eta}(u_0,u_1)\,.
  \end{align*}
  noticing that the support function of $\tilde{Z}_\eta$ is obtained
  from the support function of $Z_\eta$ by swapping the coordinates.
\end{proof}

Since $Z_\eta$ is centrally symmetric with respect to $(\thf,\thf)$,
the set $\tilde Z_\eta$ can also be obtained by reflecting $Z_\eta$
with respect to the line $\{(u_0,u_1):\; u_1=1-u_0\}$.

Lemma~\ref{le:Z-of-dual} relates the symmetry property of lift
zonoids to the duality principle in option pricing at maturity. This
principle traces its roots to observations by Merton~\cite{mer73},
Grabbe \cite{grab83}, McDonald and Schroder \cite{mcd:schr98}, Bates
\cite{bat91}, and Carr \cite{car94} and has been studied extensively
over the recent years, e.g.~Carr and Chesney~\cite{car:che96} and
Detemple~\cite{det01} discuss American version of duality. For a
detailed presentation of the duality principle in a general
exponential semimartingale setting and for its various applications
see Eberlein et al.~\cite{eber:pap:shir08} and the literature cited
therein. The multiasset case has been studied in Eberlein et
al.~\cite{eber:pap:shir08b}.

Consider a European call option valued at $c(k,F,\Q)$ with strike $k$
and maturity $T$ on a share represented in a risk-neutral world by a
$\Q$-price process $S_t=S_0e^{(r-q)t}\eta_t=F_t\eta_t$, i.e.\ the
share is traded in a market with deterministic risk-free interest rate
$r$ and attracts deterministic dividend-yield $q$, while $\eta_t$ is a
$\Q$-martingale. In the most common setting this call option is
related to a dual European put $\tilde p(S_0,\tilde F,\tilde\Q)$ with
strike $S_0$ and the same maturity $T$ on a share represented by the
dual $\tilde\Q$-price process $\tilde
S_t=ke^{(q-r)t}\tilde\eta_t=\tilde F_t\tilde\eta_t$, where
$\tilde\eta_t$ is a $\tilde\Q$-martingale. In other words, the dual
put is written on another (dual) share traded in the dual market with
risk-free interest rate $q$ assuming that this dual share attracts
dividend-yield $r$. By Lemma~\ref{le:Z-of-dual} we get the following
geometric proof and interpretation of the European \emph{call-put
  duality},
\begin{align*}
  \tilde p(S_0,\tilde F,\tilde\Q)
  &=e^{-qT}\E_{\tilde\Q}(S_0-\tilde F\tilde\eta)_+
   =e^{-qT}h_{Z_{\tilde\eta}}(S_0,-\tilde F)\\
  &=e^{-qT}h_{\tilde Z_\eta}(S_0,-\tilde F)\\
  &=e^{-qT}h_{Z_\eta}(-\tilde F,S_0)=e^{-qT}\E_\Q(S_0\eta-\tilde F)_+\\
  &=e^{-rT}\E_\Q(S_0e^{(r-q)T}\eta-ke^{(q-r)T+(r-q)T})_+=c(k,F,\Q)\,.
\end{align*}
By the classical call-put parity, a call-call duality can be
derived in a straightforward way. The duality for American options
(see~\cite{car:che96,det01,schr99}) can be interpreted
geometrically by reflecting the payoff set $A$ (see
Figure~\ref{fi:am-subdifferential}) at the line bisecting the
first quadrant.

\subsection{Symmetries for vanilla options}
\label{sec:symmetry}

If $Z_\eta$ is symmetric with respect to the line
$\{(u_0,u_1):\;u_0=u_1\}$ bisecting the first quadrant, i.e.
\begin{equation}
  \label{eq:tilde-z_eta=z_eta-}
  \tilde Z_\eta=Z_\eta\,,
\end{equation}
the duality relates out-of-the-money to certain in-the-money calls in
the same market.
We call a positive random variable $\eta$ \emph{self-dual} if
(\ref{eq:tilde-z_eta=z_eta-}) holds. This symmetry property
(\ref{eq:tilde-z_eta=z_eta-}) clearly depends on the probability
measure used to define the expectation. It implies that $\E\eta=1$,
i.e.\ in the one period setting each probability measure that makes
$\eta$ self-dual is a martingale measure.

\begin{theorem}
  \label{thr:parity}
  Assume that the asset price at maturity of an asset with
  deterministic dividend payments is $S_T=F\eta$ with self-dual
  $\eta$. Then
  \begin{equation}
    \label{eq:cf-k+k=ck-f+f}
    e^{rT}c(k,F)+k=e^{rT}c(F,k)+F\,.
  \end{equation}
\end{theorem}
\begin{proof}
  Noticing (\ref{eq:tilde-z_eta=z_eta-}), we have
  \begin{displaymath}
    e^{rT}c(k,F)=h_{Z_\eta}(-k,F)
    =h_{\tilde Z_\eta}(F,-k)=h_{Z_\eta}(F,-k)=e^{rT}p(F,k)\,.
  \end{displaymath}
  The proof is finished by applying the put-call parity.
\end{proof}

Since $c(k,F)$ is positively homogeneous in both arguments, i.e.\
$c(tk,tF)=tc(k,F)$ for $t\geq0$, the symmetry
relation~(\ref{eq:cf-k+k=ck-f+f}) can also be written as
\begin{displaymath}
  \frac{\tilde{c}(m)+e^{-rT}}{\tilde{c}(m^{-1})+e^{-rT}}=m\,,
\end{displaymath}
where $m=F/k$ determines the \emph{moneyness} of the call and
$\tilde{c}(m)=c(1,m)$.

\begin{cor}
  \label{co:ad-sym}
  Assuming (\ref{eq:tilde-z_eta=z_eta-}), the following European
  symmetries hold
  \begin{align}
    \label{eq:e-c-p-parity}
    p(k,F) &= c(F,k)\,,\\
    \label{eq:p-p-parity}
    e^{rT}p(k,F)+F &= k+e^{rT}p(F,k)\,.
  \end{align}
\end{cor}
\begin{proof}
  By~(\ref{eq:cf-k+k=ck-f+f}) and~(\ref{eq:c-p-parity}) we obtain
  \begin{displaymath}
    c(F,k)=c(k,F)-Fe^{-rT}+ke^{-rT}=p(k,F)\,.
  \end{displaymath}
  By combining the left-hand side of~(\ref{eq:cf-k+k=ck-f+f})
  with~(\ref{eq:e-c-p-parity}) for the reversed order of $k$ and $F$
  and the right-hand side of~(\ref{eq:cf-k+k=ck-f+f})
  with~(\ref{eq:e-c-p-parity}) we arrive at (\ref{eq:p-p-parity}).
\end{proof}

The above relations are known in the literature as \emph{put-call
symmetry}, see e.g.~\cite{bat97}, \cite{car94},
\cite{car:ell:gup98}, and more
recently~\cite{faj:mor06,faj:mor06b} for log-infinitely-divisible
models, which are further discussed in
Section~\ref{sec:expon-self-dual-1}. Further recent developments
are presented in~\cite{car:lee08}. There are various applications
of the put-call symmetry, especially in connection with hedging
exotic options, see~\cite{car:ell:gup98,car:lee08} and
Section~\ref{sec:ex-apl-self-dual}.

\subsection{General symmetry}
\label{sec:general-symmetry}

The following result obtained in~\cite[Th.~2.2]{car:lee08} (without
use of lift zonoids) generalises the self-duality to a wide
range of payoff functions.

\begin{theorem}
  \label{th:gen-sym}
  An integrable random variable $\eta$ is
  self-dual if and only if for any payoff function
  $f:\R_+\mapsto\R$ such that $\EQ|f(F\eta)|<\infty$ for $F>0$,
  \begin{equation}
    \label{eq:gen-sym}
    \EQ f(F\eta)=\EQ[f(F\eta^{-1})\eta]\,.
  \end{equation}
\end{theorem}
\begin{proof}
  Changing measure from $\Q$ to its dual $\tilde\Q$, we arrive at
  \begin{displaymath}
    \E_\Q f(F\eta)=\E_{\tilde\Q}[f(F\eta)\eta^{-1}]
     =\E_{\tilde\Q}[f(F\tilde\eta^{-1})\tilde\eta]\\
    =\E_\Q[f(F\eta^{-1})\eta]\,.
  \end{displaymath}
  Since lift zonoids uniquely determine distributions of random
  variables, the last equality holds by Lemma~\ref{le:Z-of-dual} in
  view of the symmetry assumption~(\ref{eq:tilde-z_eta=z_eta-}).
  Conversely, if~(\ref{eq:gen-sym}) holds for any integrable
  payoff-function, then it holds for any vanilla options,
  i.e.\ $h_{Z_\eta}(-k,F)=h_{Z_\eta}(F,-k)=h_{\tilde Z_\eta}(-k,F)$ for
  every $k>0$, implying~(\ref{eq:tilde-z_eta=z_eta-}).
\end{proof}

Denote by $\BC(k_c,F)$ and $\GC(k_c,F)$ the arbitrage-free values of a
binary call and gap call with maturity $T$ and strike $k_c$, i.e.\ the
European derivatives with payoffs $\one_{S_T>k_c}$ and
$S_T\one_{S_T>k_c}$. Furthermore, $\BP(k_p,F)$ and $\GP(k_p,F)$ denote
the arbitrage-free values of the equally specified puts, i.e.\ the
European derivatives with payoffs $\one_{S_T<k_p}$ and
$S_T\one_{S_T<k_p}$.  Theorem~\ref{th:gen-sym} yields the following
result, which is equivalent to~\cite[Cor.~2.9]{car:lee08} being a
generalisation of the \emph{binary put-call symmetry}
from~\cite{car:ell:gup98}.

\begin{cor}
  \label{co:bin-sym}
  Under the assumptions of Theorem~\ref{th:gen-sym} the following
  relationships hold:
  \begin{displaymath}
    \sqrt{k_c}\BC(k_c,F)=\frac{1}{\sqrt{k_p}}\GP(k_p,F),\qquad
    \sqrt{k_p}\BP(k_p,F)=\frac{1}{\sqrt{k_c}}\GC(k_c,F)\,,
  \end{displaymath}
  where the forward price $F$ equals the geometric mean of the binary
  (resp.\ gap) call strike $k_c$ and the gap (resp.\ binary) put
  strike $k_p$, i.e.\ $F=\sqrt{k_ck_p}$.
\end{cor}
\begin{proof}
  If $\sqrt{k_c k_p}=F$, then (\ref{eq:gen-sym}) yields that
  \begin{align*}
    e^{rT}\BC(k_c,F)&=\EQ[\one_{F\eta>k_c}]
    =\EQ[\eta\one_{F\eta^{-1}>k_c}]
     =\EQ\Big[\eta\frac{F}{\sqrt{k_ck_p}}\one_{\sqrt{k_ck_p}>k_c\eta}\Big]\\
    &=\frac{1}{\sqrt{k_ck_p}}\EQ[S_T\one_{k_p>S_T}]
     =\frac{e^{rT}}{\sqrt{k_ck_p}}\GP(k_p,F)\,.
  \end{align*}
  The proof of the second identity is similar.
\end{proof}

Thus, symmetry properties for particular options (vanilla, binary/gap,
straddles, etc.)  are nothing but writing down
Equation~(\ref{eq:gen-sym}) for special payoff functions.

\begin{figure}[htbp]
  \begin{center}
    \includegraphics[height=6.5cm]{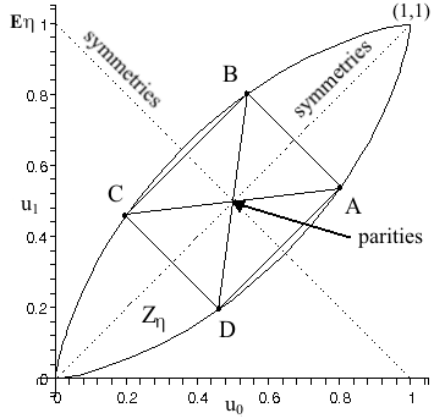}
  \end{center}
  \caption{Symmetries of $Z_\eta$ for self-dual $\eta$ and their
    financial interpretations.
    \label{fi:Z-sym}}
\end{figure}

The random variable $\eta$ has no atoms if and only if the support
function $h_{Z_\eta}$ is continuously differentiable on
$\R^2\setminus\{0\}$, so that $Z_\eta$ is strictly convex,
see~\cite[Sec.~2.5]{schn}. Then the unique point on the boundary
of $Z_\eta$ (without the points $(0,0)$, $(1,1)$) at which
$(u_0,u_1)\in\R^2\setminus\{0\}$ is the outward normal vector to
$Z_\eta$ is given by
\begin{displaymath}
  \grad h_{Z_\eta}(u_0,u_1)=
  \left(\frac{\partial h_{Z_\eta}}{\partial u_0},
        \frac{\partial h_{Z_\eta}}{\partial u_1}\right)\,.
\end{displaymath}
Assuming that the underlying probability measure $\Q$ is a
martingale measure and calculating the gradient of the support
function $h_{Z_\eta}$ at $(u_0,u_1)=(-k,F)$ with $k,F>0$ yield
another parametrisation of the upper boundary of $Z_\eta$  as
\begin{displaymath}
  \grad h_{Z_\eta}(-k,F)
  =e^{rT}\left(\BC(k,F),\frac{\GC(k,F)}{F}\right),\quad k>0\,.
\end{displaymath}
Analogously, the lower boundary is parametrised by
\begin{displaymath}
  \grad h_{Z_\eta}(k,-F)
  =e^{rT}\left(\BP(k,F),\frac{\GP(k,F)}{F}\right),\quad k>0\,.
\end{displaymath}
Hence, in non-atomic cases the boundaries of the lift zonoid
$Z_\eta$ can be parametrised by non-discounted arbitrage-free
values of binary and normalised gap options. Thus, the
distribution of $\eta$ is also reflected in these pairs of
options.

Figure~\ref{fi:Z-sym} interprets financial parities and symmetries
geometrically for the lift-zonoid $Z_\eta$.  By comparing the
coordinates of the points A, B, C, D we get relations between
binary and gap options. By comparing the related values of the
support function $h_{Z_\eta}$ we arrive at the put-call parity and
symmetry for vanilla options. Combining A with C yields parities
between certain in-the-money puts (vanilla, binary, and gap) and
the related out-of-the-money calls. Connecting the points B and D
yields the same parity results, except that B represents certain
in-the-money calls and D the related out-of-the-money puts.
Comparing C with D yields out-of-the-money put call symmetry
results, while linking A with B leads to the same results for
in-the-money options.  Finally combining B with C (resp.\ A with
D) results in the call-call (resp.\ put-put) symmetry.

\section{Options on maximum and lift max-zonoids}
\label{sec:option-prices-norms}

Since the call payoff can be written as
\begin{equation}
  \label{eq:max}
  (S_T-k)_+=\max(S_T,k)-k\,,\quad S_T,k\geq0\,,
\end{equation}
the expected call payoff can be related to another convex compact
subset $M_\eta$ of $\R^2$ that has the support function
\begin{equation}
  \label{eq:sup-rep}
  h_{M_\eta}(k,F)=\E \max(k,F\eta,0)\,,\quad (k,F)\in\R^2\,.
\end{equation}
The set $M_\eta$ is defined as the Aumann expectation of the
random triangle with vertices at the origin, $(1,0)$, and
$(0,\eta)$.  Because the financial quantities are non-negative, we
often restrict the support function onto the first quadrant
$\R_+^2$. Then it is possible to write~(\ref{eq:sup-rep}) as
\begin{equation}
  \label{eq:sup-rep-pos}
  h_{M_\eta}(k,F)=\E\max(k,F\eta)\,,\quad (k,F)\in\R_+^2\,.
\end{equation}

If $\eta=(\eta_1,\dots,\eta_n)$ is a random vector in
$\R_+^n=[0,\infty)^n$, a similar to (\ref{eq:sup-rep}) construction
leads to the set $M_\eta$ called the \emph{lift max-zonoid} of $\eta$
and defined as the Aumann expectation of the random crosspolytope with
vertices at the origin and the unit basis vectors $e_0,e_1,\dots,e_n$
in $\R^{n+1}$ scaled respectively by $1,\eta_1,\dots,\eta_n$, i.e.
\begin{displaymath}
  h_{M_\eta}(u_0,u_1,\dots,u_n)
  =\E\max(0,u_0,u_1\eta_1,\dots,u_n\eta_n)\,,
  (u_0,u_1,\dots,u_n)\in\R^{n+1}\,.
\end{displaymath}
Max-zonoids have been introduced in \cite{mo08e} in view of their
use in extreme values theory.  If $\E\eta=(1,\dots,1)$, then
$M_\eta$ is a convex compact subset of the unit cube $[0,1]^{n+1}$
that contains the origin and all unit basis vectors. If $n=1$,
then each such set is a lift max-zonoid of some random variable,
while this no longer holds for two and more assets,
see~\cite[Th.~2]{mo08e}.

\begin{theorem}
  \label{thr:uniq-lmz}
  The lift max-zonoid $M_\eta$ of an integrable random vector
  $\eta\in\R_+^n$ determines uniquely the distribution of $\eta$.
\end{theorem}
\begin{proof}
  The support function $h_{M_\eta}(u_0,u_1,\dots,u_n)$ for
  $u_0,u_1,\dots,u_n\geq0$ can be written as $\E\max(u_0,\zeta)$ for
  $\zeta=\max(u_1\eta_1,\dots,u_n\eta_n)$ and so determines uniquely
  the distribution of $\zeta$. The cumulative distribution
  function of $\zeta$ is given by
  \begin{displaymath}
    F_\zeta(t)=\P\{\eta_1\leq \frac{t}{u_1},\dots,\eta_n\leq\frac{t}{u_n}\}
  \end{displaymath}
  and so determines uniquely the joint cumulative distribution
  function of $\eta_1,\dots,\eta_n$.
\end{proof}

If the underlying probability measure is a martingale measure,
Theorem~\ref{thr:uniq-lmz} implies that prices of options on the
maxima of weighted assets determine the joint distribution of the
risky assets and so prices of all other payoffs. In view of the
positive homogeneity of support functions it suffices that the
expected values are known for parameter vectors $(u_0,u)$ with norm
one in $\R^{n+1}$, $u_0>0$ and $u$ with strictly positive
coordinates. Alternatively, it suffices to know the expected values
for fixed $u_0>0$ and $u$ with strictly positive coordinates.

The remainder of this section deals with the single asset case.
Equation (\ref{eq:max}) suggests that in this case, the lift
max-zonoid is closely related to the lift zonoid $Z_\eta$ of a
random variable $\eta$.

\begin{lemma}
  \label{le:k-and-z}
  If $M_\eta$ is the lift max-zonoid generated by a non-negative
  integrable random variable $\eta$, then
  \begin{displaymath}
    M_\eta=\conv \{(0,0)\cup (Z'_\eta+(1,0))\}\,,
  \end{displaymath}
  where $Z'_\eta$ is the reflection of $Z_\eta$ with respect to the
  line $\{(u_0,u_1):\;u_0=0\}$ and $\conv$ denotes the convex hull,
  see Figure~\ref{fi:Z-K}.
\end{lemma}
\begin{proof}
  For $(u_0,u_1)\in\R^2$ the support function of the set in the
  right-hand side is given by
  \begin{align*}
    \max(0,h_{Z'_\eta}(u_0,u_1)+u_0)&=\max(0,h_{Z_\eta}(-u_0,u_1)+u_0)\\
    &=\max(0,\E(u_1\eta-u_0)_++u_0)\,.
  \end{align*}
  By checking all possible signs of $u_0$ and $u_1$ it is easy to see
  that this support function equals
  $\E\max(0,u_0,u_1\eta)=h_{M_\eta}(u_0,u_1)$.
\end{proof}

\begin{figure}[htbp]
  \begin{center}
   \includegraphics[height=6.5cm]{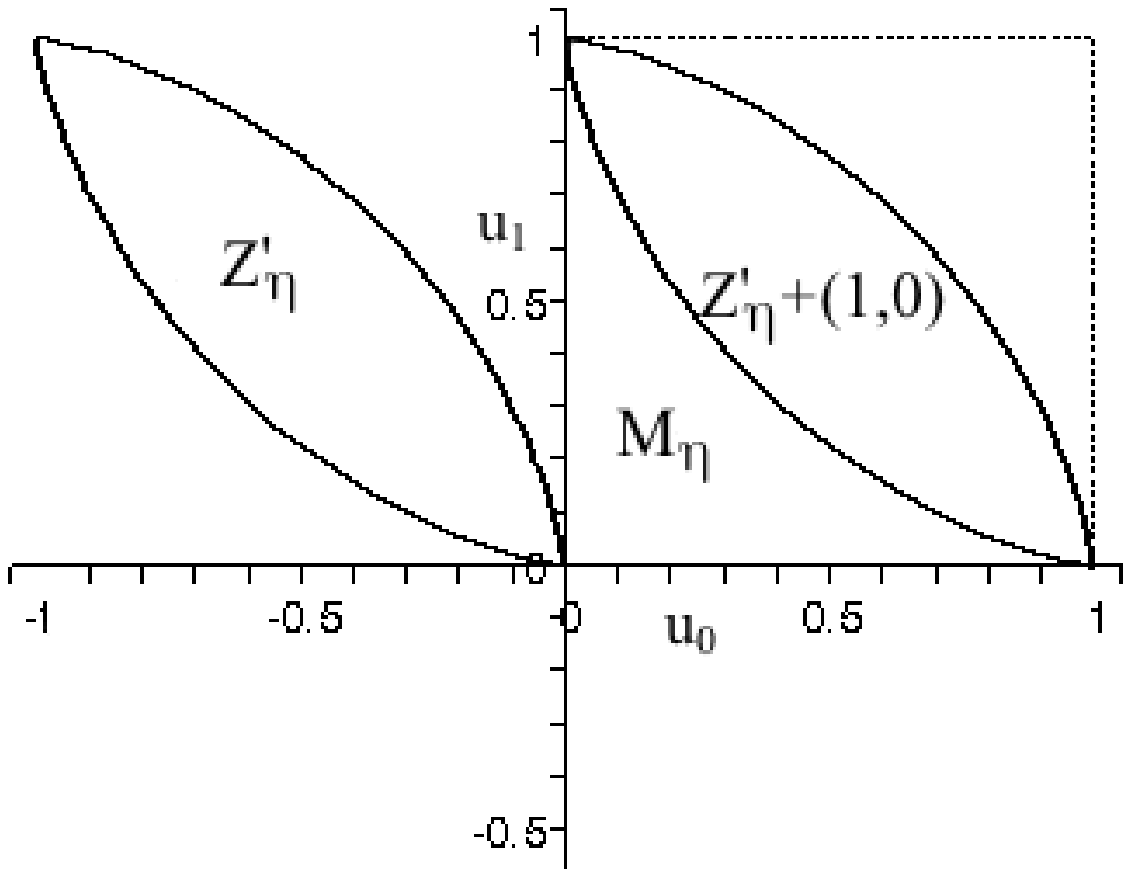}
   \includegraphics[height=6.5cm]{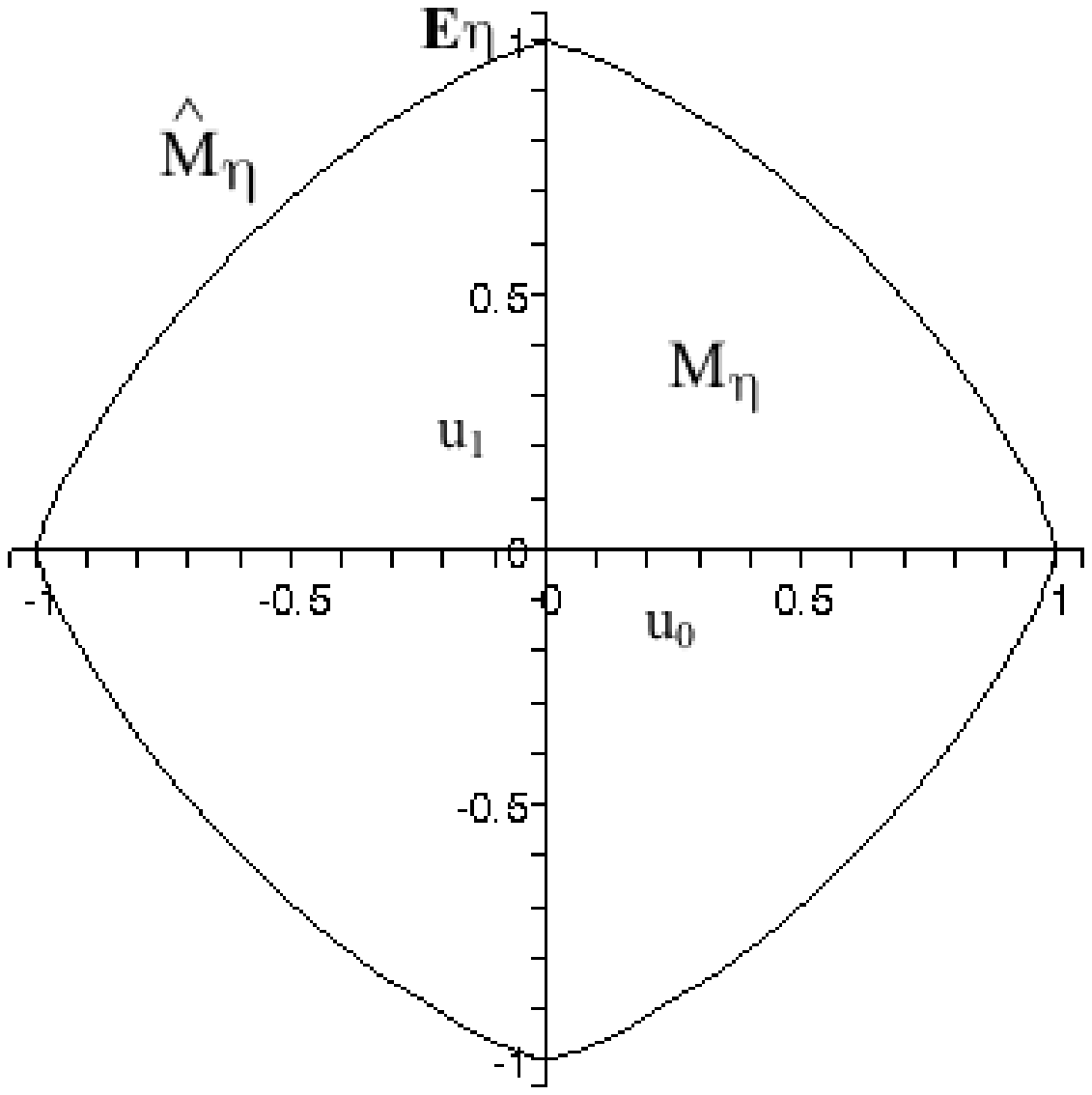}
  \end{center}
  \caption{Relation between $M_\eta$ and $Z_\eta$ as well as the body $\hat
    M_\eta$ for log-normal $\eta$ with
    mean one and volatility $\sigma=0.5$ calculated for $T=1$.
    \label{fi:Z-K}}
\end{figure}

Lemma~\ref{le:k-and-z} implies that the duality transform from
Section~\ref{sec:duality} amounts to the symmetry of $M_\eta$ with
respect to the line bisecting the first quadrant. Furthermore,
$\eta$ is self-dual if and only if $M_\eta$ is symmetric with
respect to this line, i.e.
\begin{equation}
  \label{eq:bs-sym}
  \E \max(F\eta,k)=\E\max(F,k\eta)\,, \quad k,F\geq 0\,.
\end{equation}
Indeed, since
\begin{align*}
  h_{Z_\eta}(-u_0,u_1)&=\E\max(u_1\eta,u_0)-u_0\,,\\
  h_{Z_\eta}(u_1,-u_0)&=\E\max(u_1,u_0\eta)-u_0\,,
\end{align*}
for $u_0,u_1\geq0$, the symmetry property
(\ref{eq:tilde-z_eta=z_eta-}) is equivalent to~(\ref{eq:bs-sym}).

The Euclidean space $\R^2$ can be equipped with various norms.
Each norm on $\R^2$ can be described as the support function of a
centrally symmetric (with respect to the origin) convex body that
contains the origin in its interior. Although $M_\eta$ in
Lemma~\ref{eq:sup-rep} is a subset of $\R_+^2$ and so does not
contain the origin in its interior, it is possible to use
$h_{M_\eta}$ to define the \emph{norm} on the whole plane as
\begin{displaymath}
  \|x\|_\eta=h_{M_\eta}(|x|)=h_{\hat M_\eta}(x)\,,
\end{displaymath}
where $|x|$ is the vector composed of the absolute values of the
components of $x\in\R^2$ and $\hat M_\eta$ is obtained as the
union of symmetrical transforms of $M_\eta$ with respect to the
coordinate lines, see Figure~\ref{fi:Z-K}. For instance in the
martingale setting the call price satisfies
\begin{displaymath}
  e^{rT}c(k,F)=\|x\|_\eta-k\,,\quad x=(k,F)\in\R^2_+\,.
\end{displaymath}
Conversely, each norm on $\R_+^2$ determines uniquely the
distribution of an integrable non-negative random variable.

Note that $\eta$ is self-dual if and only if the norm
$\|\cdot\|_\eta$ is symmetric, i.e.\
$\|(u_0,u_1)\|_\eta=\|(u_1,u_0)\|_\eta$ for all
$(u_0,u_1)\in\R^2_+$.

\begin{example}
  \label{eg:e-norm}
  Consider the $\ell_p$-norm on $\R^2$, which is clearly symmetric.
  Evaluating the $\ell_p$-norm of $(t,1)$, we arrive at
  \begin{displaymath}
    (t^p+1)^{1/p}=\E \max(t,\eta)
    =t\P(\eta\leq t)+\int_t^\infty xp_\eta(x)\,dx\,, \quad t>0\,,
  \end{displaymath}
  assuming that $\eta$ is absolutely continuous with density $p_\eta$.
  Differentiating with respect to $t$ yields that
  \begin{displaymath}
    \P(\eta\leq t)=t^{p-1}(t^p+1)^{1/p-1}\,.
  \end{displaymath}
  Thus, $\eta$ has the density
  \begin{displaymath}
    p_\eta(t)=(p-1)t^{p-2}(t^p+1)^{1/p-2}\,,\quad t>0\,,
  \end{displaymath}
  which is shown to imply the self-duality of $\eta$, see
  Corollary~\ref{cor:univ}(a).
\end{example}

We conclude this section by stating a relation between norms and
\emph{extreme values}. It is known~\cite{falk06,mo08e} that each norm
$\|\cdot\|$ on $\R_+^2$ corresponds to a bivariate max-stable random
vector $(\xi_1,\xi_2)$ with unit Fr\'echet marginals, i.e.
\begin{displaymath}
  \P(\xi_1\leq u_1^{-1},\;\xi_2\leq u_2^{-1})
  =\exp\{-\|(u_1,u_2)\|\}\,,\quad (u_1,u_2)\in\R_+^2\,.
\end{displaymath}
An important norm on $\R^2$ related to the Black-Scholes formula and
the theory of extreme values is mentioned in Example~\ref{eg:B-S}.

\section{Multivariate symmetry}
\label{sec:mult-symm}

\subsection{Characterisation of distributions with symmetry properties}
\label{sec:char-distr-with}

It has been shown in Section~\ref{sec:symm-prop-lift} that geometry of
the lift zonoid for a single asset price has a fundamental financial
importance.  The lifting operation amounts to adding an extra
coordinate to the asset price or prices and so increases the dimension
by one. For instance, the lift zonoid associated with the integrable
price of a single asset is a subset of the plane, which is always
centrally symmetric convex and compact.  It is well known \cite{schn}
that all centrally symmetric planar compact convex sets are zonoids,
while this is not the case in dimension 3 and more. This fact already
suggests an important dimensional effect that appears when dealing
with more than one asset.

For the multiasset case the direct relationship between lift zonoids
and lift max-zonoids is also lost. Indeed, the maximum of two numbers
can be related to the stop-loss transform as $\max(a,b)=(a-b)_++b$,
while this no longer holds for the maximum of three numbers. The
family of multivariate symmetries is also considerably richer than in
the planar case.

Let $S_T=(S_{T1},\dots,S_{Tn})$ and integrable
$\eta=(\eta_1,\dots,\eta_n)$ be as defined in~(\ref{eq:det-eta}).
Assume that all coordinates of $\eta$ are positive, i.e.\
$\eta\in\EE^n=(0,\infty)^n$ and $\eta=e^\xi$ for a random vector
$\xi=(\xi_1,\dots,\xi_n)$, where the exponential function is applied
coordinatewisely. For simplicity of notation, we do not write time
$T$ as a subscript of $\eta$ and incorporate the forward prices
$F_j$, $j=1,\dots,n$, into payoff functions, i.e.\ payoffs will be
real-valued functions of $\eta$.

In the sequel two specific payoff functions are of particular
importance, namely
\begin{displaymath}
  f_b(u_0,u_1,\dots,u_n)
  =\Big(\sum_{l=1}^n u_l\eta_l+u_0\Big)_+,\quad
   u_0, u_1,\dots,u_n\in\R\,,
\end{displaymath}
for a European basket option and
\begin{align*}
  f_m(u_0,u_1,\dots,u_n)
  =u_0\vee\bigvee_{l=1}^nu_l\eta_l\,,\qquad
  u_0, u_1,\dots,u_n\geq 0\,,
\end{align*}
for a European derivative on the maximum of $n$ weighted risky
assets together with a riskless bond, where $\vee$ denotes the
maximum operation. Despite the fact that the payoff functions $f_b$
and $f_m$ depend on $\eta$, we stress their dependence on the
coefficients, since it is crucial for symmetry properties.

By~(\ref{eq:max}), call and put options on the maximum of several
assets can be written by means of the payoff function $f_m$, e.g.
\begin{displaymath}
  \Big(\bigvee_{l=1}^nu_l\eta_l - k\Big)_+=f_m(k,u_1,\dots,u_n)-k\,,
  \quad k\geq 0\,.
\end{displaymath}
In view of Theorem~\ref{thr:uniq-lmz}, prices of these options
uniquely characterise the distribution of an integrable random
vector $\eta\in\R^n_+$.

If $X$ is the segment that joins the origin in $\R^{n+1}$ and
$(1,\eta)$, then $f_b$ becomes the support function of $X$, so that
the expected payoff is the support function
$h_{Z_\eta}(u_0,u_1,\dots,u_n)$ of the lift zonoid $Z_\eta$, i.e.
\begin{displaymath}
  \E f_b(u_0,u_1,\dots,u_n)= h_{Z_\eta}(u_0,u_1,\dots,u_n)\,.
\end{displaymath}
Similarly, for $u_0,u_1,\dots,u_n\geq 0$, the expectation of $f_m$
becomes the support function of the lift max-zonoid $M_\eta$.

Fix an arbitrary asset number $i\in\{1,\dots,n\}$ and assume that $\Q$
is a probability measure that makes $\eta$ integrable. Recall that
$\E$ without subscript denotes the expectation with respect to $\Q$,
otherwise the subscript is used to indicate the relevant probability
measure.  Since
$\eta^{1/2}=(\eta_1^{1/2},\dots,\eta_n^{1/2})=e^{\thf\xi}$ is
integrable, we can define new probability measures $\Q^i$ and $\sE^i$
by
\begin{displaymath}
  \frac{d\Q^i}{d\Q}=\frac{\eta_i}{\EQ\eta_i}\,,\qquad
  \frac{d\sE^i}{d\Q}=\frac{e^{\thf\xi_i}}{\EQ
  e^{\thf\xi_i}}\,.
\end{displaymath}
Hence, $\sE^i$ is the \emph{Esscher} (exponential) transform of $\Q$
with parameter $\thf e_i$, where $e_i$ is the $i$th standard basis
vector in $\R^n$, see~\cite{sat90} and~\cite[Ex.~7.3]{sat00} for the
Esscher transform in the context of multivariate L\'evy processes.
Since $\E_{\sE^i}e^{-\thf\xi_i}=(\EQ e^{\thf\xi_i})^{-1}$, we see
that $\Q$ is the Esscher transform of $\sE^i$ with parameter $-\thf
e_i$.

For simplicity of notation, define families of functions
$\kappa_i:\EE^n\mapsto\EE^n$ and linear mappings $K_i:\R^n\mapsto\R^n$
acting as
\begin{align}
  \kappa_i(x)&=\left(\frac{x_1}{x_i},\dots,\frac{x_{i-1}}{x_i},
    \frac{1}{x_i},\frac{x_{i+1}}{x_i},\dots,\frac{x_n}{x_i}\right)\,,\nonumber \\
  \label{eq:def-Ki}
  K_i x&=(x_1-x_i,\dots,x_{i-1}-x_i,-x_i,x_{i+1}-x_i,
  \dots,x_n-x_i)\,
\end{align}
for $i=1,\dots,n$. The linear mapping $K_i$ can be represented by
the matrix $K_i=(k_{lm})_{lm=1}^n$ with $k_{ll}=1$ for all $l\neq
i$, $k_{li}=-1$ for $l=1,\dots,n$ with all remaining entries being
$0$. Note that $\kappa_i$ and $K_i$ are self-inverse, i.e.\
$\kappa_i(\kappa_i(x))=x$ and $K_iK_ix=x$, and that the $i$th
coordinate of $K_ix$ is $-x_i$. The transpose of $K_i$ is denoted
by $K_i^\top$. In the following we consider vectors as rows or
columns depending on the situation.

The permutation of the zero-coordinate with the $i$th coordinate of
a vector $(u_0,u)\in\R^{n+1}$ is denoted by
\begin{displaymath}
  \pi_i(u_0,u)=(u_i,u_1,\dots,u_{i-1},u_0,u_{i+1},\dots,u_n)\quad
   \text{ for } i=1,\dots,n\,.
\end{displaymath}
If $B\subset\R^{n+1}$, then $\pi_i(B)$ is the reflection of $B$ at
the hyperplane $\{(u_0,u)\in\R^{n+1} : u_i=u_0\}$.

Finally, $\varphi_\xi^\Q$ (resp. $\varphi_\xi^{\sE^i}$) denotes
the characteristic function of the random vector $\xi$ under the
probability measure $\Q$ (resp. $\sE^i$).

Univariate versions of the statements (i),~(iii),~(vi), and (vii)
of the following theorem are already known from~\cite[Th.~2.2,
Cor.~2.5]{car:lee08}.

\begin{theorem}
  \label{th:mult-eq}
  Let $\eta=e^\xi$ be an $n$-dimensional $\Q$-integrable random vector
  with positive components and let $i$ be a fixed number from
  $\{1,\dots,n\}$.  The following conditions are equivalent.
  \begin{itemize}
  \item[(i)] \label{ae:basket}
    For all $u_0\in\R$ and $u\in\R^n$,
    \begin{displaymath}
      \EQ f_b(u_0,u) = \EQ f_b(\pi_i(u_0,u))\,.
    \end{displaymath}
  \item[(ii)] For all $u_0\geq 0$ and $u\in\R_+^n$,
    \begin{displaymath}
      \EQ f_m(u_0,u) = \EQ f_m(\pi_i(u_0,u))\,.
    \end{displaymath}
  \item[(iii)] For any payoff function $f:\EE^n\mapsto\R$ such that
    $\EQ |f(\eta)|<\infty$ we have
    \begin{displaymath}
      \EQ f(\eta) =\EQ[f(\kappa_i(\eta))\eta_i]\,.
    \end{displaymath}
  \item[(iv)] The lift zonoid $Z_\eta$ of $\eta$ satisfies
    $\pi_i(Z_\eta)=Z_\eta$.
  \item[(v)] The lift max-zonoid $M_\eta$ of $\eta$ satisfies
    $\pi_i(M_\eta)=M_\eta$.
  \item[(vi)] The distribution of $\eta$ under $\Q$ is identical to
    the distribution of $\tilde\eta^i=\kappa_i(\eta)$ under $\Q^i$.
  \item[(vii)] The distributions of $\xi$ and $K_i\xi$ under $\sE^i$
    coincide.
  \item[(viii)] For every $u\in\R^n$,
    \begin{displaymath}
      \varphi_\xi^{\sE^i}(u)=\varphi_\xi^{\sE^i}(K_i^\top u)
    \end{displaymath}
    or, equivalently,
    \begin{displaymath}
      \varphi_\xi^\Q\Big(u-\thf \imagi e_i\Big)
      =\varphi_\xi^\Q\Big(K_i^\top u-\thf \imagi e_i\Big)\,,
    \end{displaymath}
    where $\imagi=\sqrt{-1}$ is the imaginary unit and $e_i$ is the
    $i$th standard basis vector in $\R^n$.
  \end{itemize}
\end{theorem}

Since (vi) corresponds to the duality transform in the univariate
setting (see Section~\ref{sec:duality}), we say that $\eta$ satisfying
one of the above conditions is \emph{self-dual with respect to the
  $i$th numeraire} and write shortly $\eta\in\SD_i$. If $\eta$ is
self-dual with respect to all numeraires $i=1,\dots,n$, we call $\eta$
\emph{jointly self-dual}.

\begin{remark}[Relaxing of (i) and (ii)]
  \label{rem:reduce-c}
  In view of the positive homogeneity of payoff functions $f_b$ and
  $f_m$ it suffices to impose (i) and (ii) for parameter
  vectors $(u_0,u)$ with norm one in $\R^{n+1}$, or, in~(i), for any
  fixed $u_0\neq 0$ and any $u\in\R^n$.
  Condition (ii) can be assumed only for any fixed $u_0>0$ and $u$ with strictly
  positive coordinates, i.e.\ for $u\in\EE^n$.
\end{remark}

\begin{remark}[Martingale property]
  \label{rem:m-p}
  If $\eta\in\SD_i$, then (iii) applied to $f$ identically equal one
  (or symmetry conditions (iv), (v), (vi) for lift (max-) zonoids) imply
  that $\EQ\eta_i=1$, i.e.\ $\Q$ is a martingale measure for the $i$th
  component of $\eta$ in the one-period setting. However, $\Q$ does not need
  to be a martingale measure for other components of $\eta$, quite differently
  from the univariate case \cite{car:lee08}. The martingale property for all
  components is ensured by requiring that $\eta$ is jointly self-dual. Otherwise
  it has to be imposed additionally, if needed.
\end{remark}

\begin{remark}
  \label{rem:sd-forward}
  If the forward prices are not included in the payoff function, then
  condition (iii) for the self-duality of $\eta$ with respect to the
  $i$th numeraire can be equivalently expressed as
  \begin{multline*}
    \EQ f(S_{T1},\dots,S_{Tn})=\EQ f(F\circ\eta)=\EQ[f(F\circ\kappa_i(\eta))\eta_i]\\
    =\EQ
    \Big[f\big(\frac{S_{T1}F_i}{S_{Ti}},\dots,\frac{S_{T(i-1)}F_i}{S_{Ti}},
    \frac{(F_i)^2}{S_{Ti}},\frac{S_{T(i+1)}F_i}{S_{Ti}},\dots,\frac{S_{Tn}F_i}{S_{Ti}}\big)
    \frac{S_{Ti}}{F_i}\Big]\,,
  \end{multline*}
  by applying (iii) to $\tilde{f}(\eta)=f(F\circ\eta)$.
\end{remark}

\begin{remark}[Conditioning in Theorem~\ref{th:mult-eq}]
  \label{rem:tau-T-mult}
  All conditions of Theorem~\ref{th:mult-eq} can be written
  conditionally on a fixed event or conditionally on a
  $\sigma$-algebra.  This has the following application for stochastic
  processes.  Consider a family $\{\eta(t),\,t\geq0\}$ of random
  vectors being self-dual with respect to the $i$th numeraire. If
  $\tau$ is a non-negative random variable, which is independent of
  $\{\eta(t), \,t\geq0\}$ then $\eta(\tau)$ satisfies all statements
  in Theorem~\ref{th:mult-eq} with the expectations taken
  conditionally on the $\sigma$-algebra generated by $\tau$,
  i.e.\ $\eta(\tau)$ is conditionally self-dual with respect to the
  $i$th numeraire.
\end{remark}

To prove Theorem~\ref{th:mult-eq} we need the following multivariate
extension of the duality principle at maturity by reflection, see
Lemma~\ref{le:Z-of-dual}. The \emph{dual} lift zonoid
$Z_{\tilde\eta^i}$ (resp. lift max-zonoid $M_{\tilde\eta^i}$) with
respect to the $i$th numeraire $\eta_i$ is defined for the random
vector
$\tilde\eta^i=(\tilde\eta_1^i,\dots,\tilde\eta_n^i)=\kappa_i(\eta)$
and the expectation with respect to $\Q^i$.

\begin{lemma}
  \label{le:dual-bodies}
  Let $\eta=(\eta_1,\dots,\eta_n)$ be an integrable random
  vector with $\E\eta_i=1$ for a fixed $i\in\{1,\dots,n\}$.
  Then 
  \begin{displaymath}
    Z_{\tilde\eta^i}=\pi_i(Z_\eta)
    \quad \text{and}\quad M_{\tilde\eta^i}=\pi_i(M_\eta)\,.
  \end{displaymath}
\end{lemma}
\begin{proof}
  For $(u_0,u)\in\R^{n+1}$ we have
  \begin{align*}
    h_{Z_{\tilde\eta^i}}(u_0,u)
    &=\E_{\Q^i}\Big(\sum_{l=1}^nu_l\tilde\eta_l^i+u_0\Big)_+
    =\E_{\Q^i}\Big(\sum_{l=1,\; l\neq i}^n u_l\frac{\eta_l}{\eta_i}
      +\frac{u_i}{\eta_i}+u_0\Big)_+\\
    &=\E \left[\Big(\sum_{l=1,\; l\neq i}^n u_l\frac{\eta_l}{\eta_i}
        +\frac{u_i}{\eta_i}+u_0\Big)_+\eta_i\right]\\
    &=\E \Big(\sum_{l=1,\; l\neq i}^n u_l\eta_l
      +u_i+u_0\eta_i\Big)_+\\
    &=h_{Z_\eta}(\pi_i(u_0,u))=h_{\pi_i(Z_\eta)}(u_0,u)\,.
  \end{align*}
  The proof for lift max-zonoids is similar.
\end{proof}

\begin{proof}[Proof of Th.~\ref{th:mult-eq}]
  We will establish all equivalences in several steps.\smallskip

  \noindent
  (i)$\Rightarrow$(iv)$\Rightarrow$(vi) The definition of the lift
  zonoid, (i) (also implying $\E\eta_i=1$), and Lemma~\ref{le:dual-bodies}
  imply that for all $(u_0,u)\in\R^{n+1}$
  \begin{align*}
    h_{Z_\eta}(u_0,u)&=\EQ f_b(u_0,u)
    =\EQ f_b(\pi_i(u_0,u))\\
    &=h_{Z_\eta}(\pi_i(u_0,u))
    =h_{\pi_i(Z_\eta)}(u_0,u)
    =h_{Z_{\tilde\eta^i}}(u_0,u)\,,
  \end{align*}
  so that $Z_\eta$ and $Z_{\tilde\eta^i}$ coincide (see (iv)) as
  having identical support functions. Since the lift zonoid uniquely
  determines the distribution of an integrable random vector, this
  implies (vi). \smallskip

  \noindent
  (vi)$\Rightarrow$(iii)$\Rightarrow$(i)
  The definition of $\Q^i$, the self-inverse property of $\kappa_i$
  and (vi) yield that
  \begin{displaymath}
    \E_\Q f(\eta)=\E_{\Q^i}[f(\eta)\eta_i^{-1}]
    =\E_{\Q^i} [f(\kappa_i(\tilde\eta^i))\tilde\eta_i^i]
    =\E_\Q[f(\kappa_i(\eta))\eta_i]\,,
  \end{displaymath}
  so that (iii) holds. By applying~(iii) for the payoff function
  $f_b$ of a basket option we arrive at (i). \smallskip

  \noindent
  (iii)$\Rightarrow$(ii)$\Rightarrow$(v)$\Rightarrow$(vi)$\Rightarrow$(iii)
  By applying~(iii) to the payoff function $f_m$ we obtain (ii).
  Now the definition of the lift max-zonoid,
  (ii), and Lemma~\ref{le:dual-bodies}
  yield that
  \begin{displaymath}
    h_{M_\eta}(u_0,u)=h_{M_\eta}(\pi_i(u_0,u))=h_{\pi_i(M_\eta)}(u_0,u)
     =h_{M_{\tilde\eta^i}}(u_0,u)\,
  \end{displaymath}
  for every $(u_0,u)\in\EE^{n+1}$, and thus, for every $(u_0,u)\in\R^{n+1}$,
  i.e.\ (v) and~(vi) hold, since the lift max-zonoid uniquely
  identifies the distribution, by Theorem~\ref{thr:uniq-lmz}.
  It is already shown that (vi) implies
  (iii). \smallskip

  \noindent
  (iii)$\Leftrightarrow$(vii)$\Leftrightarrow$(viii) Since
  $m_i=\EQ(e^{\thf\xi_i})$ is finite,
  \begin{align*}
    \EQ f(\eta)&=\EQ f(e^\xi)=m_i\E_{\sE^i}[f(e^\xi)e^{-\thf\xi_i}]\,,\\
    \EQ[f(\kappa_i(\eta))\eta_i]
     &=\EQ [f(e^{K_i\xi})e^{\xi_i}]
     =m_i\E_{\sE^i}[f(e^{K_i\xi})e^{\thf\xi_i}]\,.
  \end{align*}
  Thus, (iii) yields that
  \begin{equation}
    \label{eq:e_se-thfx-e_se}
    \E_{\sE^i}[f(e^\xi)e^{-\thf\xi_i}]=
    \E_{\sE^i}[f(e^{K_i\xi})e^{\thf\xi_i}]
    =\E_{\sE^i}[f(e^{K_i\xi})e^{-\thf(K_i\xi)_i}]
  \end{equation}
  for any $\Q$-integrable payoff function $f$. Recall that the $i$th
  coordinate of $K_i\xi$ is $-\xi_i$. Choosing $f(x)=g(x)e^{\thf x_i}$,
  we see that the $\sE^i$-expectations of $g(e^\xi)$ and $g(e^{K_i\xi})$
  coincide for all continuous functions $g$ with bounded support,
  whence $\xi$ coincides in distribution with $K_i\xi$ under $\sE^i$.
  Conversely, if $\xi$ and $K_i\xi$ share the same distribution under
  $\sE^i$, then (\ref{eq:e_se-thfx-e_se}) holds and implies (iii).

  Furthermore, (vii) is equivalent to
  \begin{displaymath}
    \varphi_\xi^{\sE^i}(u)=\varphi_{K_i\xi}^{\sE^i}(u)
    =\varphi_\xi^{\sE^i}(K_i^\top u)
  \end{displaymath}
  for every $u\in\R^n$. Writing the characteristic functions as
  $\sE^i$-expectations and referring to the change of measure, the
  latter condition is equivalent to
  \begin{displaymath}
    \EQ\Big[e^{\imagi\langle u,\xi\rangle}\frac{e^{\thf\xi_i}}{\EQ
      e^{\thf\xi_i}}\Big]
    =\EQ\Big[e^{\imagi\langle K_i^\top
      u,\xi\rangle}\frac{e^{\thf\xi_i}}{\EQ e^{\thf\xi_i}}\Big]\,,
    \quad u\in\R^n\,,
  \end{displaymath}
  so that
  \begin{displaymath}
    \varphi_\xi^\Q\Big(u-\thf \imagi e_i\Big)
    =\varphi_\xi^\Q\Big(K_i^\top u-\thf \imagi e_i\Big)
  \end{displaymath}
  for all $u\in\R^n$.
\end{proof}

In view of Theorem~\ref{th:mult-eq}(iv,v), examples of random
vectors $\eta\in\SD_i$ can be derived by constructing lift (max-)
zonoids which are symmetric with respect to the hyperplane
$\{(u_0,\dots,u_n)\in\R^{n+1}:\; u_0=u_i\}$, see
Examples~\ref{eg:e-norm} and~\ref{ex:mult-unit-ball}. The following
result can be helpful for such constructions. Its univariate version
is stated in~\cite[Ex.~8]{schr99}.

\begin{theorem}
  \label{co:mult-dens}
  Consider an integrable random vector $\eta\in\EE^n$ with distribution
  $\Q$.
  \begin{itemize}
  \item[(a)] If $\eta$ is absolutely continuous with probability
    density $p_\eta$, then $\eta\in\SD_i$ if and only if
    \begin{equation}
      \label{eq:mult-eta-dens}
      p_\eta(x)=x_i^{-(n+2)}p_\eta(\kappa_i(x))
      \qquad \text{ for almost all } x\in\EE^n\,,
    \end{equation}
    equivalently, the density $p_\xi$ of $\xi=\log\eta$ satisfies
    \begin{equation}
      \label{eq:mult-xi-dens}
      p_\xi(x)=e^{-x_i}p_\xi(K_i x)\quad
      \text{ for almost all } x\in\R^n\,.
    \end{equation}
  \item[(b)] If $\eta$ is discrete, then $\eta\in\SD_i$ if and only if
    $\Q(\eta=\kappa_i(x))=x_i\Q(\eta=x)$ for each atom $x$ of $\eta$.
  \end{itemize}
\end{theorem}
\begin{proof}
  (a) Condition (iii) of Theorem~\ref{th:mult-eq} can be written in
  the integral form as
  \begin{displaymath}
    \int_{\R_+^n}f(x)p_\eta(x)\,dx
    =\int_{\R_+^n}f(\kappa_i(y))y_ip_\eta(y)\,dy
    =\int_{\R_+^n}f(x)\frac{1}{x_i^{n+2}}p_\eta(\kappa_i(x))\,dx\,,
  \end{displaymath}
  where the last equality is obtained by changing variables
  $x=\kappa_i(y)$ and noticing that $\kappa_i(\kappa_i(x))=x$.

  Consider the function $f(x)=\one_{x\in [a_1,b_1]\times\dots\times
    [a_n,b_n]}$ for any parameters \linebreak
  $a_1,\dots,a_n,b_1,\dots,b_n\in\R$.  Differentiating both sides with
  respect to $b_1,\dots,b_n$ and by using the dominated convergence,
  we get~(\ref{eq:mult-eta-dens}) almost everywhere.  For the
  converse, write the right-hand side of~(iii) as integral, refer to
  (\ref{eq:mult-eta-dens}) and change variables.  The equivalence
  between~(\ref{eq:mult-eta-dens}) and~(\ref{eq:mult-xi-dens}) can be
  seen by the classical density transformation.

  \noindent
  (b) In the discrete case~(iii) can be written as
  \begin{equation}
    \label{eq:proof-mult-disc}
    \sum_{l}f(x^l)\Q(\eta=x^l)=\sum_{l}f(\kappa_i (x^l))x^l_i\Q(\eta=x^l)\,,
  \end{equation}
  where the sum stretches over all atoms $x^l$ with $x^l_i$ being the
  $i$th component of $x^l$. Since~(\ref{eq:proof-mult-disc}) also
  holds for $f(x^l)=\one_{x=x^l}$ for any fixed atom $x$, we obtain
  that $x^*=\kappa_i(x)$ is also an atom with
  \begin{displaymath}
    \Q(\eta=x)=x^*_i\Q(\eta=x^*)
    =\frac{1}{x_i}\Q(\eta=\kappa_i(x))\,.
  \end{displaymath}
  For the converse we have for any $(u_0,u)\in\R^{n+1}$
  \begin{align*}
    h_{Z_\eta}(u_0,u)
    &=\sum_{l}\Big(u_0+\sum_{m=1}^nu_mx^l_m\Big)_+\Q(\eta=x^l)\\
    &=\sum_{l}\Big(u_0+\sum_{m=1}^nu_mx^l_m\Big)_+\;
    \frac{1}{x^l_i}\Q(\eta=\kappa_i(x^l))\\
    &=\sum_{l^*}\Big(u_i+u_0x^{l^*}_i
    +\sum_{m=1,\; m\neq i}^nu_m x^{l^*}_m\Big)_+
    \Q(\eta=x^{l^*})\\
    &=h_{Z_\eta}(\pi_i(u_0,u))\,,
  \end{align*}
  where $x^{l^*}=\kappa_i(x^l)$, so that we obtain (i) of
  Theorem~\ref{th:mult-eq}.
\end{proof}

Now we give a result about the marginal distribution of $\eta_i$ for
the random vector $\eta$ being self-dual with respect to this
numeraire.

\begin{lemma}
  \label{le:self-dual-marg}
  If $\eta\in\SD_i$, then $\eta_i$ is a self-dual random variable.
\end{lemma}
\begin{proof}
  Choose in Theorem~\ref{th:mult-eq}(ii) vector $u$ with all
  coordinates being zero apart from $u_0$ and $u_i$. Then (ii) reads
  $\EQ(u_0\vee u_i\eta_i)=\EQ(u_i\vee u_0\eta_i)$ for every
  $u_0,u_i\geq0$. This is exactly the symmetry condition
  (\ref{eq:bs-sym}).
\end{proof}

\subsection{Jointly self-dual random vectors}
\label{sec:self-dual-random}

Recall that random vector $\eta$ is called \emph{jointly self-dual}
if it is self-dual with respect to all numeraires.  Since
permutations of coordinate $0$ and an arbitrary $i\in\{1,\dots,n\}$
generate by successive applications the transpositions of any two
$i,j\in\{0,1,\dots,n\}$, the expected payoff functions $f_b$ and
$f_m$ for jointly self-dual $\eta$ are invariant with respect to any
permutation of their arguments, e.g.
\begin{equation}
  \label{eq:full-sym-in-basket}
  \E f_b(u_0,u_1,\dots,u_n)=\E f_b(u_{l_0},u_{l_1},\dots,u_{l_n})
\end{equation}
for each permutation $i\mapsto l_i$. In view of this,
Theorem~\ref{th:mult-eq} implies the following result.

\begin{theorem}
  \label{thr:sda}
  Random vector $\eta$ is jointly self-dual if and only if its lift
  (respectively lift max-) zonoid $Z_\eta$ (respectively $M_\eta$) is
  symmetric with respect to each hyperplane
  $\{(u_0,u_1,\dots,u_n)\in\R^{n+1}:\; u_i=u_j\}$ for all
  $i,j=0,\dots,n$, $i\neq j$.
\end{theorem}

\begin{cor}
  \label{cor:comp}
  If $\eta$ is jointly self-dual, then all its components are identically
  distributed self-dual random variables with expectation one and
  $\eta$ is exchangeable, i.e.\ its distribution does not change after
  any permutation of its coordinates.
\end{cor}
\begin{proof}
  The components of $\eta$ are self-dual by
  Lemma~\ref{le:self-dual-marg} and so have expectation one.  If
  $\eta$ is jointly self-dual, then Theorem~\ref{thr:sda} yields that for
  every $(u_0,u)\in\R_+^{n+1}$ and any $i,j=1,\dots,n$,
  \begin{displaymath}
    \EQ\Big(u_0\vee u_i\eta_i\vee
     \bigvee_{l=1,\; l\neq i}^n u_l\eta_l\Big)
     =\EQ\Big(u_0\vee u_i\eta_j\vee u_j\eta_i\vee
     \bigvee_{l=1,\; l\neq i,j}^n u_l\eta_l\Big)\,.
   \end{displaymath}
   By setting $u_l=0$ for all $l\neq 0,i$ we arrive at
   \begin{displaymath}
     \EQ(u_0\vee u_i\eta_i)=\EQ(u_0\vee u_i\eta_j)
     \quad\text{ for every } (u_0,u_i)\in\R_+^2\,.
   \end{displaymath}
   Thus, for any $i,j=1,\dots,n$ the random variables $\eta_i$ and
   $\eta_j$ have the same lift max-zonoid. By Theorem~\ref{thr:uniq-lmz},
   all coordinates of $\eta$ share the same distribution.

   Theorem~\ref{thr:sda} yields that $\eta=(\eta_1,\dots,\eta_n)$ and
   $(\eta_{l_1},\dots,\eta_{l_n})$ obtained by any permutation of its
   coordinates share the same lift max-zonoid and thus the same
   distribution, i.e.\ $\eta$ is exchangeable.
\end{proof}

It should be noted that the converse statement to
Corollary~\ref{cor:comp} does not hold, i.e.\ the exchangeability of
$\eta$ does not imply joint self-duality. This is easily seen as a
consequence of the following result, which says that any non-trivial
random vector $\eta$ with independent coordinates cannot be jointly
self-dual.

\begin{theorem}
  \label{thr:triv}
  Assume that $n\geq2$.
  \begin{itemize}
  \item[(a)] If $\eta\in\SD_i$ and $\eta_i$ and $\eta_j$ are
    independent for some $j\neq i$, then $\eta_i$ equals 1 almost
    surely.
  \item[(b)] If $\eta$ is a jointly self-dual random vector with independent
    coordinates, then all coordinates of $\eta$ are deterministic and
    equal 1 almost surely.
  \end{itemize}
\end{theorem}
\begin{proof}
  It suffices to prove only (a). By Theorem~\ref{th:mult-eq}(ii)
  letting $u_l=0$ for $l\neq 0,i,j$,
  \begin{displaymath}
    \EQ\big(u_0\vee u_i\eta_i\vee u_j\eta_j\big)
    =\EQ\big(u_i\vee u_0\eta_i\vee u_j\eta_j\big)
  \end{displaymath}
  for all $u_0,u_i,u_j\in\R_+$. In particular, if $u_i=0$, then
  \begin{displaymath}
    \EQ(u_0\vee u_j\eta_j)=\EQ(u_0\eta_i\vee u_j\eta_j)\quad
    \text{ for all } (u_0,u_j)\in\R_+^2\,.
  \end{displaymath}
  Since $\eta_i$ is self-dual by Lemma~\ref{le:self-dual-marg}, the
  conditioning on $\eta_j$ yields that $\E(u_0\eta_i\vee
  u_j\eta_j)=\E(u_0\vee u_j\eta_i\eta_j)$. Hence,
  \begin{displaymath}
    \EQ(u_0\vee u_j\eta_j)=\EQ(u_0\vee u_j\eta_i\eta_j)
    \quad\text{ for all } (u_0,u_j)\in\R_+^2\,,
  \end{displaymath}
  whence $\eta_j$ coincides in distribution with $\eta_i\eta_j$, see
  Theorem~\ref{thr:uniq-lmz}.  If $\xi_l=\log\eta_l$ for $l=i,j$, then
  $\xi_j$ and $\xi_i+\xi_j$ share the same distribution. Therefore,
  the characteristic function of $\xi_i$ identically equals one for
  some neighbourhood of the origin, whence $\xi_i=0$ almost surely.
\end{proof}

\begin{remark}[Random vectors sampled from L\'evy processes]
  \label{rem:levy-time}
  Assume that $\zeta_1,\dots,\zeta_n$ are independent integrable random
  variables.  Consider the vector
  \begin{displaymath}
    \eta=(\zeta_1,\zeta_1\zeta_2,\zeta_1\zeta_2\zeta_3,\dots,\zeta_1\cdots\zeta_n)\,.
  \end{displaymath}
  This construction is important, since then $\xi=\log\eta$ is a
  vector whose components form a random walk.
  However, $\eta$ cannot be jointly self-dual as a vector, unless in the
  trivial deterministic case. Indeed, setting $u_0,u_1,u_2\geq 0$ and
  $u_3=\cdots=u_n=0$, writing the expected payoff $f_m$ conditionally
  on $\zeta_2$ and using the self-duality of $\zeta_1$ (which follows
  from $\eta\in\SD_1$)
  we see that
  \begin{displaymath}
    \EQ(u_0\vee u_1\eta_1\vee u_2\eta_2)
    =\EQ(u_0\zeta_1\vee u_1\vee u_2\zeta_2)
  \end{displaymath}
  is symmetric in $u_0,u_1,u_2$ if $\eta$ is jointly self-dual. Thus,
  $(\zeta_1,\zeta_2)$ is a jointly self-dual vector with independent
  components, which is necessarily trivial by Theorem~\ref{thr:triv}.
  An extension of this argument shows that $\zeta_1=\cdots=\zeta_n=1$
  almost surely.  Therefore, it is not possible to obtain
  jointly self-dual random vectors by taking exponentials of the values
  of a L\'evy process at different time points.
\end{remark}

\begin{example}
  \label{ex:mult-unit-ball}
  The most obvious convex body in $\R^{n+1}$ being symmetric with respect
  to the hyperplanes $u_0=u_i$, $i=1,\dots,n$, is the closed unit ball
  $B_1(0)$ of radius one centred at the origin. The value of the corresponding
  derivative equals the discounted magnitude of the weight vector. It is
  shown in~\cite{mo08e} that $M_\eta=B_1(0)\cap \R_+^{n+1}$ is a max-zonoid.
  By Theorem~\ref{thr:sda}, $M_\eta$ is the lift max-zonoid of a
  jointly self-dual random vector $\eta$, such that for all
  $(u_0,u)\in\R_+^{n+1}$
  \begin{align*}
    h_{M_\eta}(u_0,u)=\|(u_0,u)\|&=\E\Big(u_0\vee\bigvee_{l=1}^n u_l\eta_l\Big)
    =u_0+\int_{u_0}^\infty \P\Big(\bigvee_{l=1}^nu_l\eta_l>t\Big)\,dt\\
    &=u_0+\int_{u_0}^\infty\Big(1-F_\eta\big(\frac{t}{u_1},\dots,\frac{t}{u_n}\big)\Big)\,dt\,,
  \end{align*}
  where $F_\eta$ is the joint cumulative distribution function of
  $\eta$.  Using the expression for the Euclidean norm $\|(u_0,u)\|$,
  differentiating with respect to all components and setting $u_0=1$
  yield the following expression for the density of $\eta$
  \begin{displaymath}
    p_\eta(u)=\frac{2^n\Gamma(n+\thf)}
    {\sqrt\pi\big(1+\sum_{l=1}^n u_l^{-2}\big)^{\thf+n}\prod_{l=1}^n u_l^3}\,,
    \quad u=(u_1,\dots,u_n)\in\EE^n\,,
  \end{displaymath}
  where $\Gamma(\cdot)$ denotes the Gamma function. It is easy to
  check that $p_\eta$
  satisfies~(\ref{eq:mult-eta-dens}) for every $i=1,\dots,n$.
\end{example}

\begin{example}
  \label{example:sda}
  Let $\zeta_0,\zeta_1,\dots,\zeta_n$ be i.i.d.\ self-dual random
  variables (their examples are provided in
  Section~\ref{sec:symm-distr}). Define $\eta_i=\zeta_0\zeta_i$ for
  $i=1,\dots,n$. Conditioning on $\zeta_1,\dots,\zeta_n$ yields that
  \begin{align*}
    \E(u_0\vee u_1\eta_1\vee\cdots \vee u_n\eta_n)
    &=\E\big[\E(u_0\vee\zeta_0(u_1\zeta_1\vee\cdots\vee u_n\zeta_n)
    |\zeta_1,\dots,\zeta_n)\big]\\
    &= \E(u_0\zeta_0\vee u_1\zeta_1\vee\cdots\vee u_n\zeta_n)
  \end{align*}
  is symmetric in $u_0,u_1,\dots,u_n$, i.e.\ $\eta=(\eta_1,\dots,\eta_n)$
  is jointly self-dual. Note that $\eta_1,\dots,\eta_n$ are all self-dual random
  variables, but are no longer independent. In particular if the $\zeta$'s
  are log-normally distributed with $\mu=-\thf$ and $\sigma=1$, then
  $\log\eta$ is normally distributed with mean $(-\thf,\dots,-\thf)$
  and the covariance matrix having diagonal elements one and all other
  $\thf$. We will return to this situation in
  Example~\ref{ex:log-n-mult}.
\end{example}

\subsection{Exponentially self-dual infinitely divisible random
  vectors}
\label{sec:expon-self-dual-1}

A random vector $\xi$ has an infinitely divisible distribution if
and only if $\xi=L_1$ for a L\'evy process $L_t$, $t\geq0$,
see~\cite{sat99}. In view of the widespread use of L\'evy models
for derivative pricing we aim to characterise infinitely divisible
random vectors $\xi=\log\eta$ for $\eta$ being self-dual with
respect to the $i$th numeraire or all numeraires.  If
$\eta\in\SD_i$, then $\xi$ is said to be \emph{exponentially
self-dual} with respect to the $i$th numeraire and we write
shortly $\xi\in\ESD_i$.

The Euclidean norm $\|\cdot\|$ is not invariant with respect to the
transformation $x\mapsto K_i x$ defined by (\ref{eq:def-Ki}).  For
simplifying the formulation of the results we introduce the following
norm on $\R^n$
\begin{equation}
  \label{eq:new-norm}
  \tn u\tn^2 =\thf(\|u\|^2+\|K_iu\|^2)\,,
  \quad u\in\R^n\,,
\end{equation}
where the number $i\in\{1,\dots,n\}$ is \emph{fixed} in the sequel.
It is easy to see that $\tnc$ is indeed a norm, which is equivalent
to the Euclidean norm on $\R^n$.  Since $K_i$ is self-inverse, $\tn
u\tn=\tn K_i u\tn$ for every $u\in\R^n$.

We use the following formulation of the \emph{L\'evy-Khintchine
  formula}, see~\cite[Ch.~2]{sat99}, for the characteristic function of
$\xi$
\begin{multline}
  \label{eq:levy-k}
  \varphi_\xi^\Q(u)=\E e^{\imagi\langle u,\xi\rangle}=
  \exp\bigg\{\imagi\langle\gammaQ,u\rangle-\thf\langle u,Au\rangle\\
  +\int_{\R^n}(e^{\imagi\langle u,x\rangle}-1
  -\imagi\langle u,x\rangle\one_{\tn x\tn\leq 1})d\nuQ(x)\bigg\}\,,
\end{multline}
for $u\in\R^n$, where $A$ is a symmetric non-negative definite
$n\times n$ matrix, $\gammaQ\in\R^n$ is a constant vector and
$\nuQ$ is a measure on $\R^n$ (called the L\'evy measure)
satisfying $\nuQ(\{0\})=0$ and
\begin{equation}
  \label{eq:nu-conditions}
  \quad\int_{\R^n}\min(\|x\|^2,1)d\nuQ(x)<\infty\,.
\end{equation}
Note that the latter condition can be equivalently written in the
new norm $\tnc$.

\begin{theorem}
  \label{th:inv-div-mult}
  Let $\eta$ be an integrable random vector under probability measure
  $\Q$ such that $\xi=\log\eta$ is infinitely divisible under $\Q$.
  Then $\xi\in\ESD_i$ if an only if for the generating triplet
  $(A,\nuQ,\gammaQ)$ the following three conditions hold.
  \begin{itemize}
  \item[(1)] The matrix $A=(a_{lj})_{lj=1}^n$ satisfies
    $a_{ij}=a_{ji}=\thf a_{ii}$ for all $j=1,\dots,n$, $j\neq i$.
  \item[(2)] The L\'evy measure satisfies
    \begin{equation}
      \label{eq:dnuqx=e-x_idn-text}
      d\nuQ(x)=e^{-x_i}d\nuQ(K_ix)\quad \text{ almost everywhere }
    \end{equation}
    meaning that $\nuQ(B)=\int_{K_iB} e^{x_i}d\nuQ(x)$ for all Borel
    $B$.
  \item[(3)] The $i$th coordinate of  $\gammaQ$ satisfies
    \begin{equation}
      \label{eq:gamm-=gamm-left}
      \gammaQ_i=\int_{\tn x\tn\leq 1}x_i(1-e^{\thf x_i})\,d\nuQ(x)-\thf a_{ii}\,.
    \end{equation}
  \end{itemize}
\end{theorem}
\begin{proof}
  Since $\eta$ is positive integrable, $0<\EQ e^{\thf\xi_i}<\infty$,
  so that the Esscher transform $\sE^i$ of $\Q$ with parameter
  $\thf e_i$ and the inverse transform are well defined.
  According to \cite{sat90} or~\cite[Ex.~7.3]{sat00}, $\xi$ under
  $\sE^i$ has also an infinitely divisible distribution, so that
  \begin{displaymath}
    \varphi_\xi^{\sE^i}(u)=
    \exp\bigg\{\imagi\langle\gamma^{\sE^i},u\rangle-\thf\langle u,Au\rangle\\
    +\int_{\R^n}(e^{\imagi\langle u,x\rangle}-1
    -\imagi\langle u,x\rangle\one_{\tn x\tn\leq 1})d\nu^{\sE^i}(x)\bigg\}
  \end{displaymath}
  for a new vector $\gamma^{\sE^i}$ and L\'evy measure
  $\nu^{\sE^i}$. Note that the matrix $A$ is invariant under the
  Esscher transform, see~\cite{sat90} or~\cite[Ex.~7.3]{sat00}.

  By~Theorem~\ref{th:mult-eq}(viii) $\xi\in\ESD_i$ if and only if
  \begin{equation}
    \label{eq:varph-u-quadt}
    \varphi_\xi^{\sE^i}(u)=\varphi_\xi^{\sE^i}(K_i^\top u)
    \quad\text{ for all } u\in\R^n\,.
  \end{equation}
  By the L\'evy-Khintchine formula,
  \begin{multline*}
    \varphi_\xi^{\sE^i}(K_i^\top u)=
    \exp\bigg\{\imagi\langle\gamma^{\sE^i},K_i^\top u\rangle-\thf\langle
    K_i^\top u,AK_i^\top u\rangle\\
    +\int_{\R^n}(e^{\imagi\langle K_i^\top u,x\rangle}-1
    -\imagi\langle K_i^\top u,x\rangle\one_{\tn x\tn\leq 1})d\nu^{\sE^i}(x)\bigg\}\,.
  \end{multline*}
  Noticing that $\langle K_i^\top u,x\rangle=\langle u,K_ix\rangle$,
  changing the variable $x$ to $K_ix$ in the last integral, using the
  $K_i$-invariance of $\tnc$ and the self-inverse property of $K_i$,
  we see that
  \begin{multline*}
    \varphi_\xi^{\sE^i}(K_i^\top u)=
    \exp\bigg\{\imagi\langle K_i\gamma^{\sE^i},u\rangle-\thf\langle
    u,K_iAK_i^\top u\rangle\\
    +\int_{\R^n}(e^{\imagi\langle u,x\rangle}-1
    -\imagi\langle u,x\rangle\one_{\tn x\tn\leq 1})d\nu^{\sE^i}(K_ix)\bigg\}\,.
  \end{multline*}
  The uniqueness of the parameters $A$, $\nu^{\sE^i}$, and
  $\gamma^{\sE^i}$ of the L\'evy-Khintchine representation for
  $\varphi_\xi^{\sE^i}$ (see~\cite[Th.~8.1]{sat99}) implies that
  (\ref{eq:varph-u-quadt}) holds if and only if
  \begin{align}
    \label{eq:phi1}
    A&=K_iAK_i^\top\,,\\
    \label{eq:phi2}
    \gamma^{\sE^i}&=K_i\gamma^{\sE^i}\,,
  \end{align}
  and the L\'evy measure $\nu^{\sE^i}$ is $K_i$-invariant.

  Using the self-inverse property of $K_i$, representing $K_i$ as the
  difference of the unit matrix and the matrix $K'_i$ which has all
  zeroes apart from the $i$th column $1,\dots,1,2,1,\dots,1$ with $2$
  at the $i$th position, and by equating the entries in $K'_iA=A(K'_i
  )^\top$, we easily obtain that~(\ref{eq:phi1}) holds if and only if
  $a_{ij}=a_{ji}=\thf a_{ii}$ for every $j=1,\dots,n$, $j\neq i$.
  Next, (\ref{eq:phi2}) holds if and only if the $i$th component of
  $\gamma^{\sE^i}$ is $0$.

  Since the norm $\tnc$ does not change integrability properties in
  the L\'evy-Khintchine representation, it is possible to replicate
  the proof from~\cite{sat90} to show that the
  Esscher transform with parameter $ -\thf e_i$ leaves $A$ invariant
  while other parts of the L\'evy triplet are transformed as
  \begin{align*}
    d\nuQ(x)&=e^{-\thf x_i}d\nu^{\sE^i}(x)\,,\\
    \gammaQ&=\gamma^{\sE^i}
    +\int_{\tn x\tn\leq 1}x(e^{-\thf x_i}-1)d\nu^{\sE^i}(x)+A\big(-\thf e_i\big)\,.
  \end{align*}
  The latter condition is equivalent to
  (\ref{eq:gamm-=gamm-left}), noticing that the $i$th component of
  $A(\thf e_i)$ is $a_{ii}/2$, $d\nu^{\sE^i}(x)=e^{\thf x_i}d\nuQ(x)$
  and $\gamma^{\sE^i}$ has zero as its
  $i$th component, while other components are arbitrary.

  Furthermore, for almost all $x$,
  \begin{displaymath}
    d\nuQ(x)=e^{-\thf x_i}d\nu^{\sE^i}(x)
    =e^{-x_i}e^{-\thf (-x_i)}d\nu^{\sE^i}(K_ix)
    =e^{-x_i}d\nuQ(K_ix)\,,
  \end{displaymath}
  where again we used the fact that the $i$th component of $K_ix$ is
  $-x_i$ and that $\nu^{\sE^i}$ is $K_i$-invariant.

  Conversely, the integrability of $\eta$ ensures the existence of the
  Esscher transform of $\Q$ with parameter $\thf e_i$.  By doing this
  transform and the converse calculations it is easy to verify that
  Theorem~\ref{th:mult-eq}(viii) applies, i.e.\ $\eta=e^\xi\in\SD_i$.
\end{proof}

Since an infinitely divisible random variable $\xi$ is symmetric if
and only if $\gamma$ vanishes and the L\'evy measure is symmetric, the
above proof is very short in the univariate case and immediately
yields the corresponding univariate result stated
in~\cite{faj:mor06,faj:mor06b}, see also \cite{car:lee08} and
Corollary~\ref{cor:univ-inf-div}.

The $\SD_i$-property of $\eta$ implies that the $i$th component of
$\eta$ has expectation one. If this holds for other components,
e.g.\ if $\eta$ forms a martingale, this imposes further
restrictions on the coordinates of $\gammaQ$, namely
\begin{displaymath}
  \gammaQ_j+\thf a_{jj}+\int_{\R^n}(e^{x_j}-1
  -x_j\one_{\tn x\tn\leq 1})d\nuQ(x)=0\,, \quad j=1,\dots,n\,.
\end{displaymath}

\begin{remark}[The role of the norm]
  \label{rem:norm}
  If we use the Euclidean norm to define the truncation in
  (\ref{eq:levy-k}), then this change only affects the value of
  $\gammaQ$, while $A$ and $\nuQ$ remain the same.  If
  $\gammaQ_{\|\cdot\|}$ denotes the ``drift'' calculated for the
  Euclidean norm, then
  \begin{displaymath}
    \gammaQ_{\|\cdot\|} =\gammaQ
    +\int_{\R^n}x(\one_{\|x\|\leq 1}-\one_{\tn x\tn\leq
    1})d\nuQ\,,
  \end{displaymath}
  so that (\ref{eq:gamm-=gamm-left}) transforms into
  \begin{displaymath}
    \gammaQ_{\|\cdot\|,i}=\int_{\R^n}x_i\big(\one_{\|x\|\leq 1}-\one_{\tn
      x\tn\leq1}e^{\thf x_i}\big) d\nuQ(x)-\thf a_{ii}\,.
  \end{displaymath}
\end{remark}

\begin{remark}[Jointly self-dual case]
  \label{rem:scd}
  Assume that the conditions of Theorem~\ref{th:inv-div-mult} hold for
  each $i=1,\dots,n$. The first condition implies that $A$ equals up to
  a constant factor the matrix which has all $1$ on diagonal and $\thf$
  outside. By applying~(\ref{eq:dnuqx=e-x_idn-text}) consecutively to
  coordinates $i\neq j$ and noticing that $K_iK_jK_i$ defines the
  transposition of the $i$th and $j$th coordinates of $n$-dimensional
  vectors, we see that in this case the L\'evy measure $\nu$ is
  invariant under permutations and all components of $\gamma$ coincide.
\end{remark}

\begin{remark}[Finite mean case]
  \label{rem:fmc}
  Now we also assume that $\xi$ has finite mean, which is the case if
  and only if $\int_{\|x\|>1}\|x\|d\nuQ(x)<\infty$,
  see~\cite[Cor.~25.8]{sat99}. Then we can rewrite (\ref{eq:levy-k}) in
  the following form
  \begin{equation}
    \label{eq:levy-k-moment}
    \varphi_\xi^\Q(u)=\exp\left\{\imagi\langle\muQ,u\rangle
      -\thf\langle u,Au\rangle
      +\int_{\R^n}(e^{\imagi\langle u,x\rangle}-1
      -\imagi\langle u,x\rangle)d\nuQ(x)\right\}
  \end{equation}
  for $u\in\R^n$, where $\muQ$ is the $\Q$-expectation of $\xi$.
  Replicating the proof of Theorem~\ref{th:inv-div-mult} (or by
  adjusting $\gammaQ$ and using $d\nuQ(x)=e^{-x_i}d\nuQ(K_ix)$) we
  obtain that $\xi\in\ESD_i$ if and only if conditions (1) and (2) of
  Theorem~\ref{th:inv-div-mult} hold, while (\ref{eq:gamm-=gamm-left})
  is replaced by
  \begin{equation}
    \label{eq:mu-condition}
    \muQ_i=\int_{\R^n} x_i(1-e^{\thf x_i})\,d\nuQ(x)-\thf a_{ii}\,.
  \end{equation}

\end{remark}

\begin{example}[Log-normal distribution, Black-Scholes setting]
 \label{ex:log-n-mult}
 Assume that $\eta$ is log-normal with underlying normal vector
 $\xi=\log\eta$, so that
 \begin{displaymath}
   \varphi_\xi^\Q(u)=\exp\Big\{\imagi\langle\muQ,u\rangle
   -\thf\langle u,Au\rangle\Big\}\,,\quad
   u\in\R^n\,.
 \end{displaymath}
 Then $\eta\in\SD_i$ if and only if the covariance matrix
 $A=(a_{lm})_{lm=1}^n$ satisfies $a_{li}=a_{il}=\thf a_{ii}$ for
 $l=1,\dots,n$, $l\neq i$, and
 $\muQ_i=-\thf a_{ii}$, see~(\ref{eq:mu-condition}).

 Finally, $\eta$ is \emph{jointly} self-dual if and only if
 $a_{ll}=\sigma^2$ for all $l=1,\dots,n$, $a_{lm}=\thf\sigma^2$
 for all $l\neq m$, i.e.\ for $\sigma>0$ the correlations between
 $\xi_i$ and other components of $\xi$ are $\thf$, and the mean is
 $-\frac{\sigma^2}{2}$ for all $l=1,\dots,n$. The mean and covariance
 matrix of $\xi$ are
 then
 \begin{equation}
   \label{eq:mvar-normal}
   -\,\frac{\sigma^2}{2}(1,\dots,1) \quad\text{and}\quad
   \sigma^2
   \begin{pmatrix}
     1 & \thf &\cdots &\thf\\
     \thf & 1 & \cdots &\thf\\
     \vdots & \vdots & \vdots & \vdots\\
     \thf & \thf &\cdots & 1
   \end{pmatrix}\,.
 \end{equation}
\end{example}

\begin{remark}[Square integrable case and covariance]
  \label{re:korr}
  As a consequence of Example~\ref{ex:log-n-mult}, for log-normal
  $\eta=e^\xi\in\SD_i$ the correlations between the $i$th and other
  components of $\xi$ are $\thf\sqrt{a_{ii}/a_{ll}}$
  (assuming $a_{ii},a_{ll}>0$), while other
  correlations are not affected. In order to relax this
  correlation structure between the $i$th and other components,
  it is useful to introduce a jump component. For doing that, we assume
  $\int_{\|x\|>1}\|x\|^2d\nu(x)<\infty$, i.e.\ $\xi$ is
  square-integrable. Then the elements of the covariance matrix of
  $\xi$ are given by
  \begin{displaymath}
    v_{lj}=a_{lj}+\int_{\R^n}x_lx_jd\nu(x)\,,
  \end{displaymath}
  see~\cite[Ex.~25.12]{sat99}, i.e.\ despite of the constrains
  on the L\'evy measure given in~(\ref{eq:dnuqx=e-x_idn-text}) there are
  various possibilities for the covariance and correlation structures.
  Simple examples can be constructed as in the following remark,
  see also Remark~\ref{rem:l-measure} and Example~\ref{ex:levym-from-normal}
  (for $\alpha=1$).
\end{remark}

\begin{remark}[L\'evy measures]
  \label{rem:lm-sdc}
  Assume that $\xi\in\ESD_i$ is infinitely divisible with the L\'evy
  measure $\nu$.  If $\nu$ is finite, then the second condition of
  Theorem~\ref{th:inv-div-mult} means that random vector $\zeta$
  distributed according to the normalised $\nu$ is $\ESD_i$ itself.
  In particular, if $\nu$ is absolutely continuous, its density
  satisfies (\ref{eq:mult-xi-dens}). An immediate example of $\nu$
  is Gaussian law with the mean and variance from
  Example~\ref{ex:log-n-mult}, so that $e^\zeta$ is log-normally
  distributed as in Example~\ref{ex:log-n-mult}. Since this $\nu$
  is finite, the non-Gaussian part of $\xi$ corresponds to the
  compound Poisson law with Gaussian jumps.
\end{remark}

\subsection{Quasi-self-dual vectors}
\label{sec:quasi-self-dual-1}

As we have seen, the symmetry properties of random price changes
(interpretation of $\eta$ in a risk-neutral case) are considered
separately from the forward prices of the assets. In some cases,
notably for semi-static hedging of barrier options with carrying
costs, see~\cite{car:cho97,car:cho02,car:lee08}, the symmetry is
imposed on price changes adjusted with carrying costs
$a=e^\lambda\in\R^n$, where usually $\lambda_j=r-q_j$ for the
risk-free interest rate $r$ and $q_j$ is the dividend yield of the
$j$th asset, $j=1,\dots,n$.

In view of applications to derivative pricing it is natural to
assume that all components of $\eta$ have expectation one, i.e.\
$\Q$ is a one-period martingale measure. Then the random vector
\begin{displaymath}
  e^\lambda\circ\eta
  =(e^{\lambda_{1}}\eta_{1},\dots,e^{\lambda_{n}}\eta_{n})
\end{displaymath}
cannot be self-dual with respect to the $i$th numeraire (resp.\ for
all numeraires) unless $\lambda_{i}=0$ (resp.\ all components of
$\lambda$ vanish), since the multiplication by $e^\lambda$ moves the
expectation away from one. One can however relate
$e^\lambda\circ\eta$ to a self-dual random vector by means of a
power transformation.

\begin{definition}
  A random vector $\eta\in\EE^n$ is said to be \emph{quasi-self-dual}
  (of order $\alpha$) if there exist $\lambda\in\R^n$ and $\alpha\neq0$
  such that $(e^\lambda\circ \eta)^\alpha$ is integrable and self-dual with
  respect to the $i$th numeraire. We then write
  $\eta\in\QSD_i(\lambda,\alpha)$.
\end{definition}

If $\eta\in\QSD_i(\lambda,\alpha)$, then $\E
(e^{\lambda_i}\eta_{i})^\alpha=1$ by
Lemma~\ref{le:self-dual-marg}, so that the values of $\alpha$ and
$\lambda_i$ are closely related to each other.  Later in this
section, we discuss this relation for a special case of
quasi-self-dual L\'evy models. If useful, $\lambda$ can also have
other interpretations than being the pure carrying costs and one
can also drop the assumption that $\eta$ is a one-period
martingale itself. If imposed, the martingale assumption will be
explicitly mentioned.

By Theorem~\ref{th:mult-eq}(iii), $\eta\in\QSD_i(\lambda,\alpha)$
yields that
\begin{equation}
  \label{eq:quasi-sd-gen-sym-dif-cc}
  \E f(e^\lambda\circ\eta)
  =\E f\big(\big((e^\lambda\circ\eta)^\alpha\big)^\frac{1}{\alpha}\big)
  =\E[f(\kappa_i(e^\lambda\circ\eta))(a_{i}\eta_{i})^\alpha]\,.
\end{equation}

Define random vector $\zeta=\lambda+\xi$, where $\eta=e^\xi$ for
$\xi=(\xi_{1},\dots,\xi_{n})$. Then $e^\lambda\circ\eta=e^\zeta$.
If we consider the payoff function as a function of asset prices
$S_T=(S_{T1},\dots,S_{Tn})$ with $S_{Tj}=S_{0j}e^{\zeta_{j}}$ for
$j=1,\dots,n$, then (\ref{eq:quasi-sd-gen-sym-dif-cc}) can be
written as
\begin{displaymath}
  \E f(S_T)
  =\E\Big[f\Big(\frac{S_{0i}}{S_{Ti}}
  (S_{T1},\dots,S_{T(i-1)},S_{0i},S_{T(i+1)},
  \dots,S_{Tn}\big)\Big)\Big(\frac{S_{Ti}}{S_{0i}}\Big)^\alpha\Big]\,.
\end{displaymath}

Fix an asset number $i\in\{1,\dots,n\}$ and assume now that $\Q$
is a probability measure such that
$\eta\in\QSD_i(\lambda,\alpha)$.  Since $\eta^\alpha$ is positive
integrable, $0<\EQ e^{\frac{\alpha}{2}\zeta_{i}}<\infty$.
Hence, we can define probability measure $\tilde\sE^i$ by
\begin{displaymath}
  \frac{d\tilde\sE^i}{d\Q}
  =\frac{e^{\frac{\alpha}{2}\zeta_{i}}}{\EQ e^{\frac{\alpha}{2}\zeta_{i}}}\,,
  \qquad
  \frac{d\Q}{d\tilde\sE^i}
  =\frac{e^{-\frac{\alpha}{2}\zeta_{i}}}
   {\E_{\tilde\sE^i} e^{-\frac{\alpha}{2}\zeta_{i}}}\,,
\end{displaymath}
i.e.\ the Esscher transform of $\Q$ with parameter $\frac{\alpha}{2}
e_i$ and the corresponding inverse transform.

It is obvious that $\eta\in\QSD_i(\lambda,\alpha)$ is equivalent to
any of the condition of Theorem~\ref{th:mult-eq} for $(e^\lambda\circ
\eta)^\alpha=e^{\alpha\zeta}$.  The following theorem yields a more
direct characterisation.

\begin{theorem}
  \label{th:mult-eq-quasi-sd}
  Let $\eta^\alpha$ be integrable for some $\alpha\neq0$.  Then
  $\eta\in\QSD_i(\lambda,\alpha)$ is equivalent to one of the
  following conditions for $\zeta$ defined from
  $e^\zeta=e^{\lambda}\circ\eta=e^{\lambda+\xi}$.
  \begin{itemize}
  \item[(i)] For any payoff function $f:\EE^n\mapsto\R$ such that
    $\E|f(e^\zeta)|<\infty$
    \begin{equation}
      \label{eq:quasi-sd-gen-sym-zeta}
      \E f(e^\zeta)
      =\E [f(e^{K_i\zeta})e^{\alpha\zeta_{i}}]\,.
    \end{equation}
  \item[(ii)] The distributions of $\zeta$ and $K_i\zeta$ under
    $\tilde\sE^i$ coincide.
  \item[(iii)] For every $u\in\R^n$,
    \begin{displaymath}
      \varphi_{\zeta}^{\tilde\sE^i}(u)=\varphi_{\zeta}^{\tilde\sE^i}(K_i^\top u)
    \end{displaymath}
    or, equivalently,
    \begin{displaymath}
      \varphi_{\xi}^\Q\Big(u-\frac{\alpha}{2}\imagi e_i\Big)
      =\varphi_{\xi}^\Q\Big(K_i^\top u-\frac{\alpha}{2}\imagi e_i\Big)
        e^{-\imagi\lambda_i(\sum_{l=1}^n u_l+u_i)}\,.
    \end{displaymath}
  \end{itemize}
  Moreover, if additionally $\eta$ is integrable, we have that
  $\eta\in\QSD_i(\lambda,\alpha)$ if and only if
  (\ref{eq:quasi-sd-gen-sym-zeta}) holds for $f$ being payoffs from
  basket options with arbitrary strikes and weights of assets.
\end{theorem}

Note also that all conditions of Theorem~\ref{th:mult-eq-quasi-sd}
can be written conditionally on a fixed event or conditionally on a
$\sigma$-algebra, cf.\ Remark~\ref{rem:tau-T-mult}. The joint
quasi-self-duality can be achieved by raising the components of
$\eta$ to different powers.

\begin{proof}[Proof of Th.~\ref{th:mult-eq-quasi-sd}]
  For (i) it suffices to note that $\eta$ is quasi-self-dual if and
  only if $\alpha\zeta\in\ESD_i$ and refer to
  (\ref{eq:quasi-sd-gen-sym-dif-cc}) and Theorem~\ref{th:mult-eq}(iii).

  Replace $\eta$ by $e^\lambda\circ\eta$, $\EQ [f(\kappa_i(\eta))\eta_i]$ by
  $\EQ[f(\kappa_i(e^\lambda\circ\eta))(e^{\lambda_i}\eta_{i})^\alpha]$,
  $\sE^i$ by $\tilde\sE^i$, $\xi$ by $\zeta$, and $\thf$ by
  $\frac{\alpha}{2}$ in the proof
  of the equivalence (iii)$\Leftrightarrow$(vii) in
  Theorem~\ref{th:mult-eq} to see that (i) is equivalent to (ii). A similar
  argument yields the equivalence of~(ii) and
  $\varphi_{\zeta}^{\tilde\sE^i}(u)
  =\varphi_{\zeta}^{\tilde\sE^i}(K_i^\top u)$ for all $u\in\R^n$ as
  well as the equivalence of this equation with
  \begin{equation}
    \label{eq:char-f-zeta}
    \varphi_{\zeta}^\Q\Big(u-\frac{\alpha}{2}\imagi e_i\Big)
    =\varphi_{\zeta}^\Q\Big(K_i^\top u-\frac{\alpha}{2}\imagi e_i\Big)
  \end{equation}
  for all $u\in\R^n$. Writing the characteristic functions as
  $\Q$-expectations and using that $\zeta=\lambda+\xi$ yields
  that
  \begin{multline*}
    \EQ \exp\big\{\imagi\langle u-\frac{\alpha}{2}\imagi e_i,\xi\rangle
      +\imagi\langle u-\frac{\alpha}{2}\imagi e_i,\lambda \rangle\big\}
    \\=\EQ \exp\big\{\imagi\langle K_i^\top u-\frac{\alpha}{2}\imagi e_i,\xi\rangle
      +\imagi\langle K_i^\top u-\frac{\alpha}{2}\imagi e_i,\lambda \rangle\big\}\,.
  \end{multline*}
  Dividing by $\exp\{\imagi\langle
  u-\frac{\alpha}{2}\imagi e_i,\lambda \rangle\}$ yields the
  equivalence of the second statement in~(iii)
  and~(\ref{eq:char-f-zeta}).

  If for integrable $\eta=e^{\zeta-\lambda}$
  \begin{equation}
    \label{eq:quasi-on-sup-f}
    \E(u_0+\langle u,e^\zeta\rangle)_+
    =\E\big[(u_0+\langle u,\kappa_i(e^\zeta)\rangle)_+ e^{\alpha\zeta_i}\big]
  \end{equation}
  holds for every $(u_0,u)\in\R^{n+1}$ we first have that
  $\E e^{\alpha\zeta_i}=1$ by letting $u_0=1$ and
  $u_1=u_2=\dots=u_n=0$. Hence, we can define the measure $\P$ by
  \begin{displaymath}
    \frac{d\P}{d\Q}=e^{\alpha\zeta_i}\,,
  \end{displaymath}
  so that
  \begin{displaymath}
   \E(u_0+\langle u,e^\zeta\rangle)_+
    =\E\big[(u_0+\langle u,\kappa_i(e^\zeta)\rangle)_+ e^{\alpha\zeta_i}\big]
    =\E_\P(u_0+\langle u,\kappa_i(e^\zeta)\rangle)_+
  \end{displaymath}
  for every $(u_0,u)\in\R^{n+1}$, i.e., by~\cite[Th.~2.21]{mos02},
  $e^\zeta$ under $\Q$ and $\kappa_i(e^\zeta)$ under $\P$ share the
  same distribution. Hence, a payoff function is $\Q$-integrable
  if and only if $\E_\P|f(\kappa_i(e^{\zeta}))|<\infty$ and for every
  $\Q$-integrable payoff-function we have
  \begin{displaymath}
    \E f(e^\zeta)=\E_\P f(\kappa_i(e^\zeta))
    =\E [f(\kappa_i(e^\zeta))e^{\alpha\zeta_{i}}]\,,
  \end{displaymath}
  i.e.\ we arrive at~(\ref{eq:quasi-sd-gen-sym-zeta}). The other
  implication is obvious.
  \end{proof}

We now use Theorem~\ref{th:mult-eq-quasi-sd} to characterise all
quasi-self-dual $\eta$ such that $\xi=\log\eta$ is infinitely
divisible with the L\'evy-Khintchine representation~(\ref{eq:levy-k}).

\begin{theorem}
  \label{th:inv-div-mult-quasi-sd}
  Let the random vector $\xi=\log\eta$ be infinitely divisible under
  $\Q$ with the generating triplet $(A,\nuQ,\gammaQ)$ and let
  $\eta^\alpha$ be integrable for some $\alpha\neq0$. Then
  $\eta\in\QSD_i(\lambda,\alpha)$ if an only if the following three
  conditions hold.
  \begin{itemize}
  \item[(1)] The matrix $A=(a_{lj})_{lj=1}^n$ satisfies
    $a_{ij}=a_{ji}=\thf a_{ii}$ for all $j=1,\dots,n$, $j\neq i$.
  \item[(2)] The L\'evy measure satisfies
    \begin{equation}
     \label{eq:cond-2}
      d\nuQ(x)=e^{-\alpha x_i}d\nuQ(K_ix)\quad \text{ almost everywhere }
    \end{equation}
    meaning that $\nuQ(B)=\int_{K_iB} e^{\alpha x_i}d\nuQ(x)$ for all Borel
    $B$.
  \item[(3)] The $i$th coordinate of  $\gammaQ$ satisfies
    \begin{equation}
      \label{eq:gamm-=gamm-left-quasi-sd}
      \gammaQ_i
      =\int_{\tn x\tn\leq 1}x_i(1-e^{\frac{\alpha}{2} x_i})\,d\nuQ(x)
      -\frac{\alpha}{2} a_{ii}-\lambda_i\,.
    \end{equation}
  \end{itemize}
\end{theorem}
\begin{proof}
  Denote $\zeta=\lambda+\xi$. Since $0<\EQ
  e^{\frac{\alpha}{2}\zeta_{i}}<\infty$, the Esscher transform
  $\tilde\sE^i$ of $\Q$ with parameter $\frac{\alpha}{2} e_i$ and the
  inverse transform are well defined.  Therefore, $\zeta$ under
  $\tilde\sE^i$ has also an infinitely divisible distribution. By
  using Theorem~\ref{th:mult-eq-quasi-sd}(iii) instead of
  Theorem~\ref{th:mult-eq}(viii) and replacing $\sE^i$ by $\tilde\sE^i$, $\xi$
  by $\zeta$, $\thf$ by $\frac{\alpha}{2}$~(\ref{eq:levy-k})
  in the proof of Theorem~\ref{th:inv-div-mult}, we obtain~(1), (2), and
  \begin{displaymath}
    \gammaQ_i=
    \int_{\tn x\tn\leq 1} x_i(1-e^{\frac{\alpha}{2}x_i})d\nuQ(x)
    -\frac{\alpha}{2}a_{ii}
  \end{displaymath}
  for the generating triplet of $\zeta$ under $\Q$. Since
  $\xi=\zeta-\lambda$ we only have to adjust $\gammaQ_i$ by
  $-\lambda_i$ to finish the proof of the first implication.

  The integrability of $\eta^\alpha$ implies the existence of the
  Esscher transform of $\Q$ with parameter $\frac{\alpha}{2} e_i$. By
  doing this transform and the converse calculations it is easy to
  verify that Theorem~\ref{th:mult-eq-quasi-sd}(iii) applies, i.e.\
  $\eta=e^{\xi}\in\QSD_i(\lambda,\alpha)$.
\end{proof}

Note that condition~(1) is identical to
Theorem~\ref{th:inv-div-mult}(1).  As a consequence of
Theorem~\ref{th:inv-div-mult}(2)
and~\ref{th:inv-div-mult-quasi-sd}(2), we immediately get the
following results.

\begin{cor}
  \label{co:nontr-levy-m}
  Let the random vector $\xi=\log\eta$ be infinitely divisible under
  $\Q$ with non-vanishing L\'evy-measure $\nu$.  Then $\eta$ cannot be
  quasi-self-dual of two different orders with respect to the same
  numeraire.
\end{cor}

\begin{cor}
  \label{cor:all-T}
  If $\xi_t$, $t\geq0$, is the L\'evy process with generating triplet
  $(A,\nuQ,\gammaQ)$ that satisfies the conditions of
  Theorem~\ref{th:inv-div-mult-quasi-sd}, then
  $e^{\xi_t}\in\QSD_i(\lambda t,\alpha)$ for all $t\geq0$.
\end{cor}
\begin{proof}
  It suffices to note that $\phi_{\xi_t}^\Q(u)=(\phi_{\xi_1}^\Q(u))^t$
  for all $t\geq0$ and raise the corresponding identity from
  Theorem~\ref{th:mult-eq-quasi-sd}(iii) into power $t$.
\end{proof}

\begin{remark}[L\'evy measures in the quasi-self-dual case]
  \label{rem:l-measure}
  In order to construct a L\'evy measure $\nu$ satisfying
  (\ref{eq:cond-2}), note that
  \begin{displaymath}
    e^{\frac{\alpha}{2}x_i}d\nu(x)=e^{\frac{\alpha}{2}(K_ix)_i}d\nu(K_ix)\,,
  \end{displaymath}
  meaning that the measure $\nu_0$ with density
  $\frac{d\nu_0}{d\nu}(x)=e^{\frac{\alpha}{2}x_i}$ is
  $K_i$-invariant. Therefore, in the background one always needs to
  have a $K_i$-invariant L\'evy measure.

  Since the Lebesgue measure on $\R^n$ is $K_i$-invariant, a simple
  example of $\nu_0$ is provided by the Lebesgue
  measure restricted onto $B_R$, where $B_R=\{x:\tn x\tn\leq R\}$ is
  the ball of radius $R$ in the $\tnc$ norm. A further implication of
  the $K_i$-invariance property of the Lebesgue measure on $\R^n$ is that
  the Lebesgue density $p_{\nu_0}$ of an absolutely continuous
  $K_i$-invariant measure $\nu_0$ is also $K_i$-invariant, i.e.\
  $p_{\nu_0}(x)=p_{\nu_0}(K_i x)$ for almost every $x\in\R^n$.
  Then~(\ref{eq:cond-2}) can be equivalently written as
  $p_\nu(x)=e^{-\alpha x_i}p_\nu(K_ix)$.  Clearly,
  condition~(\ref{eq:nu-conditions}) is always satisfied for a finite
  $\nu$ without atom at the origin, which then yields the compound
  Poisson part of $\xi$ from $\eta=e^\xi\in\QSD_i(\lambda,\alpha)$.
  The integrability condition on $\eta^\alpha$ additionally requires
  that
  \begin{displaymath}
    \int_{\|x\|>1}
    e^{\alpha x_j}e^{-\frac{\alpha}{2}x_i}\,d\nu_0(x)<\infty\,,\quad j=1,\dots,n\,,
  \end{displaymath}
  see~\cite[Th.~25.17]{sat99}.
\end{remark}

\begin{remark}[Determining $\alpha$ from the carrying costs in the
  risk-neutral case]
  \label{rem:alpha}
  Assume that $\E\eta_j=1$ for all $j=1,\dots,n$ and
  $\eta\in\QSD_i(\lambda,\alpha)$ with given $\lambda$. Since
  $\phi_\xi^\Q(-\imagi e_j)=\E\eta_j=1$, we see that
  \begin{equation}
    \label{eq:stand-gamma}
    \gammaQ_j=-\int_{\R^n}(e^{x_j}-1-x_j\one_{\tn x\tn\leq 1})
      d\nuQ(x)-\thf a_{jj}\,,\quad j=1,\dots,n\,.
  \end{equation}
  If $\alpha=1$, then the above condition for $j=i$ yields
  (\ref{eq:gamm-=gamm-left}) (or (\ref{eq:gamm-=gamm-left-quasi-sd})
  for $\alpha=1$ and $\lambda=0$). Indeed, it suffices to check that
  \begin{multline*}
    \int_{\R^n}(1-e^{x_i}+x_ie^{\thf x_i}\one_{\tn x\tn\leq 1})\,d\nuQ(x)\\
    =\int\limits_{\{x_i<0\}}(1-e^{x_i}+x_ie^{\thf x_i}
      \one_{\tn x\tn\leq 1})e^{-x_i}\,d\nuQ(K_ix)
    +\int\limits_{\{x_i>0\}}(1-e^{x_i}+x_ie^{\thf x_i}
      \one_{\tn x\tn\leq 1})\,d\nuQ(x)\\
    =\int\limits_{\{y_i>0\}}(e^{y_i}-1-y_ie^{\thf y_i}
      \one_{\tn y\tn\leq 1})\,d\nuQ(y)
    +\int\limits_{\{x_i>0\}}(1-e^{x_i}+x_ie^{\thf x_i}
      \one_{\tn x\tn\leq 1})\,d\nuQ(x)=0\,.
  \end{multline*}
  However, for non-vanishing $\lambda$ we need to
  combine~(\ref{eq:stand-gamma}) with~(\ref{eq:gamm-=gamm-left-quasi-sd})
  to see that $\alpha$ must satisfy
   \begin{displaymath}
    -\int_{\R^n}(e^{x_i}-1-x_i\one_{\tn x\tn\leq 1})d\nuQ(x)-\thf a_{ii}
    =\int\limits_{\tn x\tn\leq 1} x_i(1-e^{\frac{\alpha}{2} x_i})\,d\nuQ(x)
    -\frac{\alpha}{2} a_{ii}-\lambda_i\,,
  \end{displaymath}
  or, equivalently,
  \begin{equation}
    \label{eq:levy-p}
    a_{ii}\alpha=a_{ii}-2\lambda_i+2
    \int_{\R^n}(e^{x_i}-1-x_ie^{\frac{\alpha}{2}x_i}\one_{\tn x\tn\leq 1})d\nuQ(x)\,.
  \end{equation}
  It should be noted that in the L\'evy processes setting from
  Corollary~\ref{cor:all-T} the values of $\alpha$ calculated for all
  $t\geq0$ coincide.
\end{remark}

\begin{remark}[Finite mean case]
  \label{rem:qsd-fm}
  Assume that $\eta=e^\xi\in\QSD_i(\lambda,\alpha)$.  If, as in
  Remark~\ref{rem:fmc}, $\xi$ has finite mean, then
  (\ref{eq:gamm-=gamm-left-quasi-sd}) is replaced by
  \begin{equation}
    \label{eq:mu-condition-quasi-sd}
    \muQ_i=\int_{\R^n} x_i(1-e^{\frac{\alpha}{2} x_i})\,d\nuQ(x)
    -\frac{\alpha}{2} a_{ii}-\lambda_i\,,
  \end{equation}
  where $\muQ$ is the expectation of $\xi$. If $\Q$ is a martingale
  measure for $\eta_i$, then $\varphi_\xi^\Q(-\imagi e_i)=\E
  e^{\xi_i}=1$ yields that
  \begin{equation}
    \label{eq:mu-i-finite-mean-c}
    \mu_i=-\int_{\R^n}(e^{x_i}-1-x_i)d\nu(x)-\thf a_{ii}\,.
  \end{equation}
  Combining~(\ref{eq:mu-condition-quasi-sd})
  with~(\ref{eq:mu-i-finite-mean-c}) yields
  \begin{align}
    a_{ii}\alpha&=a_{ii}-2\lambda_i+2
    \int_{\R^n}(e^{x_i}-1-x_ie^{\frac{\alpha}{2}x_i})d\nuQ(x)
    \nonumber\\
    \label{eq:levy-p-finit-mean}
    &=a_{ii}-2\lambda_i+2
    \int_{\R}(e^{x_i}-1-x_ie^{\frac{\alpha}{2}x_i})d\nuQ_i(x_i)\,,
  \end{align}
  where $\nuQ_i$ is the marginal L\'evy measure defined by
  $\nu_i(B)=\nu(\{x\in\R^n:\; x_i\in B\})$ for Borel $B\subset\R$,
  $0\notin B$, see~\cite[Prop.~11.10]{sat99}.

  Compared to~(\ref{eq:levy-p}), Equation~(\ref{eq:levy-p-finit-mean})
  yields a considerable simplification in calculating $\alpha$.
  Since $\nu_i$ is the L\'evy measure corresponding to $\eta_i$, it is
  possible to calculate $\alpha$ from only the distribution of the
  $i$th component of $\eta$ and the corresponding carrying costs
  $\lambda_i$.

  In the purely non-Gaussian case (i.e.\ if $A$ vanishes) it is useful
  to write the integral in (\ref{eq:levy-p-finit-mean}) as its
  principal value. Then the principal value of the integral of
  $x_ie^{\frac{\alpha}{2}x_i}$ vanishes, since
  $d\nu_i(x_i)=e^{-\frac{\alpha}{2}x_i}d\nu_{0i}(x_i)$ for a symmetric
  measure $\nu_{0i}$, and
  \begin{align*}
   \lambda_i&=\pvint (e^{x_i}-1)d\nu_i(x_i)
   =\pvint (e^{x_i}-1)e^{-\frac{\alpha}{2}x_i}d\nu_{0i}(x_i)\\
   &=\pvint (e^{(1-\frac{\alpha}{2})x_i}-e^{-\frac{\alpha}{2}x_i})d\nu_{0i}(x_i)\,.
  \end{align*}
  If $\nu_{0i}$ has a finite Laplace transform $\psi$ on the real
  line, then $\alpha$ solves
  \begin{displaymath}
    \lambda_i=\psi(1-\frac{\alpha}{2})-\psi(-\frac{\alpha}{2})\,.
  \end{displaymath}
\end{remark}

\begin{example}[Log-normal model with carrying costs]
  \label{ex:lnc-cc}
  By Corollary~\ref{co:nontr-levy-m}, among all log-infinitely
  divisible distributions only the log-normal one can be
  quasi-self-dual of two orders with respect to the
  same numeraire. Applying~(\ref{eq:levy-p}) for the univariate
  log-normal case with $a_{ii}=\sigma^2>0$ (and vanishing $\nu$) yields that
  \begin{displaymath}
    \alpha=1-\frac{2\lambda}{\sigma^2}\,,
  \end{displaymath}
  as stated in~\cite{car:cho97,car:cho02,car:lee08}. Hence, the
  univariate log-normal distribution in the Black-Scholes setting is
  self-dual and quasi-self-dual of order $1-\frac{2\lambda}{\sigma^2}$
  at the same time. By~(\ref{eq:levy-p}), this is also true for
  multivariate log-normal models from Example~\ref{ex:log-n-mult}
  being self-dual with respect to the $i$th numeraire, i.e.\ this
  distribution is at the same time quasi-self-dual of order
  $\alpha=1-2\lambda_i/a_{ii}$ with respect to the same numeraire.
\end{example}

\begin{example}[Determining $\alpha$ for non-trivial L\'evy measures]
  \label{ex:levym-from-normal}
  Start with the univariate case (i.e.\ $n=1$) and choose $\nu_0$ from
  Remark~\ref{rem:l-measure} to be the centred Gaussian measure with
  variance $\beta^2>0$.  If normalised to have the total mass one, $\nu$
  becomes the density of the normal law with mean
  $-\,\frac{\alpha\beta^2}{2}$ and variance $\beta^2$.
  Solving~(\ref{eq:levy-p}) or
  equivalently~(\ref{eq:levy-p-finit-mean}) for this particular
  measure $\nu$ and $a_{ii}=\sigma^2>0$ yields that
  \begin{displaymath}
    \alpha=\frac{1}{\beta^2\sigma^2}
    \left(2\mathrm{LambertW}\Big(\frac{\beta^2}{\sigma^2}
      \exp\big\{\frac{\beta^2(\lambda+1)}{\sigma^2}\big\}\Big)
      \sigma^2+\beta^2\sigma^2-2\beta^2\lambda-2\beta^2\right)\,,
  \end{displaymath}
  where $\mathrm{LambertW}(x)=g(x)$ is the principal branch of the
  $\mathrm{LambertW}$ function that satisfies
  $g(x)e^{g(x)}=x$ for all $x$.
  In the purely non-Gaussian case the required power is given by
  \begin{displaymath}
    \alpha=1-\frac{2}{\beta^2}\log(1+\lambda)\,.
  \end{displaymath}

  In the multivariate case we start with $\nu_0$ being the centred
  Gaussian law having positive definite covariance matrix $B$ that satisfies
  Theorem~\ref{th:inv-div-mult-quasi-sd}(1) for some fixed $i$ and
  define measure $\nu$ with density
  \begin{displaymath}
    \frac{d\nu}{d\nu_0}(x)=e^{-\frac{\alpha}{2}x_i}\,.
  \end{displaymath}
  Then~(\ref{eq:cond-2}) holds and the $i$th
  marginal $\nu_i$ of $\nu$ has the density $e^{-\frac{\alpha}{2}x_i}$
  with respect to the $i$th marginal of $\nu_0$, the latter being the
  centred Gaussian law with variance $\beta^2=b_{ii}$. Since the $i$th
  marginal for the normalised $\nu$ coincides with the L\'evy measure
  constructed above in the univariate case, we obtain the same
  $\alpha$ as in the univariate case with $\beta=\sqrt{b_{ii}}$.
\end{example}

\section{Distributions of self-dual random variables}
\label{sec:symm-distr}

\subsection{Characterisation and examples}
\label{sec:char-exampl}

In this section we specialise the results from
Section~\ref{sec:mult-symm} for studying self-dual random
\emph{variables}. Denote by $\bar{F}(x)=\P(\eta>x)$ the tail of the
cumulative distribution function of a positive random variable
$\eta$ and by
\begin{displaymath}
  \bar{F}_I(z)=\int_0^z \bar{F}(t)dt\,,\quad z\geq0\,,
\end{displaymath}
the \emph{integrated tail}. Note that $\bar{F}_I(0)=0$ and
$\bar{F}_I(\infty)=1$ in case $\E\eta=1$.

\begin{theorem}
  \label{thr:symm}
  An integrable positive random variable $\eta$ is self-dual if and
  only if $\bar{F}_I(\infty)=1$ and
  \begin{equation}
    \label{eq:gen-sym-funct}
    z\bar{F}_I(z^{-1})=\bar{F}_I(z)\quad\text{ for all } z>0\,.
  \end{equation}
\end{theorem}
\begin{proof}
  It is easy to check that
  \begin{displaymath}
    \bar{F}_I(z)=\E \min(\eta,z)\,,\quad z\geq 0\,.
  \end{displaymath}
  Now apply Theorem~\ref{th:gen-sym} or Theorem~\ref{th:mult-eq}~(iii)
  to the payoff function $f(\eta)=\min(\eta,z)$ to see that
  \begin{displaymath}
    \bar{F}_I(z)=\E\min(\eta,z)=\E[\min(\eta^{-1},z)\eta]
    =\E\min(1,z\eta)=z\bar{F}_I(z^{-1})\,.
  \end{displaymath}

  In the opposite direction, (\ref{eq:gen-sym-funct}) yields that
  \begin{displaymath}
    \E\max(\eta,z)=\E[\eta+z-\min(\eta,z)]
    =\E[1+z\eta-\min(1,z\eta)]=\E\max(1,z\eta)\,,
  \end{displaymath}
  i.e.\ by rescaling (cf.~Remark~\ref{rem:reduce-c}) we arrive at
  the self-duality property (\ref{eq:bs-sym}).
\end{proof}

Theorem~\ref{co:mult-dens} in the univariate case yields the following
result, which is known from \cite[Ex.~8]{schr99}.

\begin{cor}
  \label{cor:univ}
  Let $\eta$ be a positive integrable random variable with
  distribution $\Q$.
  \begin{itemize}
  \item[(a)] If $\eta$ is absolutely continuous with probability density
    $p_\eta$, then $\eta$ is self-dual if and only if
    \begin{equation}
      \label{eq:dens-cond}
      p_\eta(x)=x^{-3}\,p_\eta(x^{-1})\quad\text{ for almost all } x>0\,.
    \end{equation}
    If $\xi=\log\eta$, the self-duality of $\eta$ (i.e.\ the
    exponential self-duality of $\xi$) is equivalent to
    \begin{displaymath}
      p_\xi(x)=e^{-x}\, p_\xi(-x) \quad\text{ for almost all } x\in\R\,.
    \end{displaymath}
  \item[(b)] If $\eta$ has a discrete distribution, then $\eta$ is
    self-dual if and only if $\Q(\eta=x^{-1})=x\Q(\eta=x)$ for each
    atom $x$ of $\eta$.
  \end{itemize}
\end{cor}

Clearly, if the density $p_\eta$ is continuous, then
(\ref{eq:dens-cond}) holds for all $x>0$. For instance, the
probability density of the log-normal distribution of mean one
satisfies~(\ref{eq:dens-cond}). It is also satisfied by mixtures
of log-normal densities that appear in the (uncorrelated)
Hull-White stochastic volatility model,
see~\cite[Th.~3.1]{gul:stein06}. The self-duality property of
stochastic volatility models is explored in
\cite[Th.~3.1]{car:lee08}.

\begin{example}[Log-normal model]
  \label{eg:B-S}
  If $S_T=F\eta$ has the log-normal distribution, the Black-Scholes
  formula yields that
  \begin{equation}
    \label{eq:hr-norm}
    \E\max(F\eta,k)=F\Phi(d')+k\Phi(d'')\,,
  \end{equation}
  where $k,F>0$,
  \begin{displaymath}
    d'=\lambda+\frac{1}{2\lambda}\,\log\frac{F}{k}\,,\quad
    d''=\lambda-\frac{1}{2\lambda}\,\log\frac{F}{k}\,,\quad
    \lambda=\thf\sigma\sqrt T\,,
  \end{displaymath}
  and $\Phi$ is the cumulative distribution function for the standard
  normal variable. Note that the conventional Black-Scholes formula is
  obtained by subtracting $k$ from~(\ref{eq:hr-norm}) and then
  discounting. By looking at the right-hand side of~(\ref{eq:hr-norm})
  it is easy to see that it is symmetric with respect to $F$ and $k$,
  i.e.\ $\eta$ is a self-dual random variable.

  The right-hand side of~(\ref{eq:hr-norm}) defines a (symmetric) norm
  on $\R_+^2$ called the \emph{H\"usler-Reiss norm} of $x=(k,F)$,
  see~\cite{mo08e}.  Thus, the derivative given by the maximum of the
  asset price and the strike has the price given by the discounted
  norm of the vector composed of the forward and the strike.
  Notably, expression~(\ref{eq:hr-norm}) appears in the
  literature on extreme values, see~\cite{hues:reis89}, as the limit
  distribution of coordinatewise maxima for triangular arrays of
  bivariate Gaussian vectors with correlation $\varrho(n)$ that
  approaches one with rate $(1-\varrho(n))\log n\to\lambda^2\in
  [0,\infty]$ as $n\to\infty$.
\end{example}

In order to construct further examples of probability density
functions $p_\eta$ that satisfy~(\ref{eq:dens-cond}) it suffices to
define $p_\eta(x)$ for $x\geq1$ and then extend it for $x\in(0,1)$
using~(\ref{eq:dens-cond}) with a subsequent normalisation to ensure
that the total mass is one. Clearly, one has to bear in mind that
$\E\eta=1$ presumes the integrability of $xp_\eta(x)$ (alongside with
$p_\eta(x)$ itself) at zero and infinity.

\begin{example}[Self-dual random variables with heavy tails]
  \label{ex:havy-tailed}
  The log-normal distribution has a light tail at infinity. It is
  possible to construct a self-dual \emph{heavy-tail} distribution by
  setting
  \begin{equation}
    \label{eq:ex-dens}
    p(x) = \\
    \begin{cases}
      c_\gamma x^\gamma        & \text{if}\; x \in (0,1] \,,\\
      c_\gamma x^{-(3+\gamma)} & \text{if}\; x > 1\,,
    \end{cases}
  \end{equation}
  for $\gamma>-1$, where $c_\gamma=(1+\gamma)(2+\gamma)/(3+2\gamma)$
  normalises the probability density.
\end{example}

\begin{example}[Discrete self-dual random variable]
  If $\eta$ takes values $\thf,1,2$ with probabilities
  $\frac{1}{3},\thf,\frac{1}{6}$, then Corollary~\ref{cor:univ}(b)
  implies that $\eta$ is self-dual.
\end{example}

\begin{remark}
  If $\eta$ is not self-dual, then (\ref{eq:bs-sym}) is clearly
  violated, but the resulting inequalities can not be the same way
  around for every $k,F\geq0$.  Without loss of generality assume that
  $F=1$. Then $\E\max(k,\eta)\leq\E\max(k\eta,1)$ for all $k\geq0$
  with the strict inequality for some $k=k_0$ leads to a
  contradiction, since
  \begin{align*}
    \E\max(k_0,\eta)<\E\max(k_0\eta,1)
    &=k_0\E\max(\eta,k_0^{-1})\\
    &\leq k_0\E\max(k_0^{-1}\eta,1)
    =\E\max(\eta,k_0)\,.
  \end{align*}
\end{remark}

\subsection{Moments of self-dual random variables}
\label{sec:moments}

It is immediate that all self-dual random variables have expectation
one. Carr and Lee \cite[Cor.~2.6]{car:lee08} show that
\begin{equation}
  \label{eq:c-l-moments}
  \E\eta^n=\E\eta^{-n+1}\,,\quad n\geq1\,.
\end{equation}
In particular, if $\E\eta^2<\infty$, then
\begin{displaymath}
  \mathrm{Cov}(\eta^{-1},\eta)
  =1-\E\eta^{-1}=(\E\eta)^2-\E\eta^2=-\mathrm{Var}(\eta)\,.
\end{displaymath}

\begin{theorem}
  \label{thr:skew}
  Each non-trivial self-dual variable $\eta$ with finite third moment
  has a positive skewness
  $\E(\eta-\E\eta)^3/(\mathrm{Var}(\eta))^{3/2}$.
\end{theorem}
\begin{proof}
  In view of (\ref{eq:c-l-moments}),
  \begin{align*}
    \E(\eta-\E\eta)^3&=\E(\eta^3-3\E\eta^2+2)\\
    &=\E(\eta^3-3\E\eta^2+2+6(\eta-1)+\eta^{-1}-\eta^2)\\
    &=\E\Big[(\eta-1)^2(\eta+\eta^{-1}-2)\Big]\geq0\,.
  \end{align*}
  This also shows that the skewness vanishes if and only if $\eta=1$
  almost surely.\footnote{The authors thank a referee for suggesting
    the current proof.}
\end{proof}

\begin{remark}[Product of self-dual variables]
  \label{rem:prod-sd}
  If $\eta_1$ and $\eta_2$ are two independent self-dual random
  variables, then
  \begin{align*}
    \E\max(k,\eta_1\eta_2)&=\E[\E(\max(k,\eta_1\eta_2)|\eta_1)]\\
    &=\E[\E(\max(k\eta_2,\eta_1)|\eta_2)]=\E\max(k\eta_1\eta_2,1)\,,
  \end{align*}
  i.e.\ the product $\eta_1\eta_2$ is self-dual. By taking successive
  products it is possible to construct a sequence of self-dual random
  variables, whose logarithms build a random walk. Note however that
  the values of this random walk at different time points are not
  jointly self-dual, cf.\ Remark~\ref{rem:levy-time}.
\end{remark}

\subsection{Exponentially self-dual variables}
\label{sec:expon-self-dual}

Theorem~\ref{th:mult-eq}(viii) implies that $\xi$ is exponentially
self-dual if and only if the characteristic function $\phi_\xi^\Q$
satisfies
\begin{displaymath}
  \phi_\xi^\Q(u-\thf \imagi)=\phi_\xi^\Q(-u-\thf \imagi)\,,\quad
  u\in\R\,.
\end{displaymath}
If $\xi$ has an absolutely continuous distribution,
Corollary~\ref{cor:univ}(a) yields that $\xi$ is self-dual if and only
if $e^{\thf y}p_\xi(y)$ is an even function of $y$.

If $\xi$ is also infinitely divisible, then its distribution is
characterised by the L\'evy triplet $(\sigma^2,\nuQ,\gammaQ)$.  Note
that in the univariate case $K_1x=-x$, the norm (\ref{eq:new-norm})
becomes the Euclidean one and $A$ reduces to a single number
$\sigma^2$.  Theorem~\ref{th:inv-div-mult} yields the following
univariate result, known from Fajardo and
Mordecki~\cite{faj:mor06,faj:mor06b}; to see that their ``drift'' with
truncation function $\one_{|x|\leq 1}$ is equal to $\gammaQ$ from
Corollary~\ref{cor:univ-inf-div} use $e^{-x}d\nu(-x)=d\nu(x)$. The
latter condition on the L\'evy measure appears also in Carr and
Lee~\cite[Th.~4.1]{car:lee08}.

\begin{cor}
  \label{cor:univ-inf-div}
  An integrable random variable $\eta=e^\xi$ with $\xi$ being infinitely
  divisible represented by the L\'evy
  triplet $(\sigma^2,\nuQ,\gammaQ)$ is exponentially self-dual if
  and only if $d\nuQ(x)=e^{-x}d\nuQ(-x)$
  and
  \begin{equation}
    \label{eq:gammaq=int_xleq-1-x1}
    \gammaQ=\int_{|x|\leq 1} x(1-e^{\thf x})d\nuQ(x)-\frac{\sigma^2}{2}\,.
  \end{equation}
  If $\xi$ is integrable, then~(\ref{eq:gammaq=int_xleq-1-x1}) can be
  replaced by the following condition on its expectation
  \begin{displaymath}
    \muQ=\int_{\R} x(1-e^{\thf x})d\nuQ(x)-\frac{\sigma^2}{2}\,.
  \end{displaymath}
\end{cor}

While Corollary~\ref{cor:univ-inf-div} is obtained as a univariate
version of Theorem~\ref{th:inv-div-mult}, it is alternatively possible
first to describe the (univariate) dual market in terms of its
generating triplet, and then ensure that the generating triplets of
the original and dual markets coincide, implying that $\eta$ is
self-dual, see~\cite{faj:mor06}. The latter approach also describes
the dynamics of the dual market in the univariate case.

If $\E|\xi |^3<\infty$, then~\cite[Prop.~3.13]{con:tan} yields that
\begin{displaymath}
  \EQ(\xi-\E\xi)^3=\int_\R x^3d\nuQ(x)
  =\int_0^\infty x^3(1-e^x)d\nuQ(x)\,.
\end{displaymath}
Thus, the skewness of exponentially self-dual $\xi$ is negative
except in the log-normal case, where it is zero.

\subsection{Quasi-self-dual variables and asymmetry corrections}
\label{sec:quasi-self-dual}

Let $S_T=S_0a\eta$ for $S_0,a>0$ with $\eta$ being a general
positive random variable, so that the forward price is given by
$F=S_0a$. Assume that $\eta$ is absolutely continuous with
non-vanishing density $p_\eta$ and $\E\eta=1$. Then it is possible
to find a function $q_{a\eta}$ such that
\begin{align}
  \E f(S_T)=\E[f(S_0/(a\eta))q_{a\eta}(a\eta)]
  &=\E[f(F/(a^2\eta))q_{a\eta}(a\eta)]\nonumber\\
  \label{eq:qsd}
  &=\E[f((S_0)^2/S_T)q_{a\eta}(a\eta)]\,
\end{align}
for each function $f:\R_+\mapsto\R$ such that $f(S_T)$ is integrable.
Indeed, it suffices to choose
\begin{displaymath}
  q_{a\eta}(x)=\frac{p_{a\eta}(x^{-1})}{x^2p_{a\eta}(x)}
  =\frac{p_{S_T}(x^{-1}S_0)}{x^2p_{S_T}(xS_0)}\,.
\end{displaymath}
By choosing $x=a\eta=S_T/S_0$ we arrive at
\begin{displaymath}
  \E f(S_T)=\E\Big[f\big(\frac{(S_0)^2}{S_T}\big)
  \big(\frac{S_T}{S_0}\big)^{-2}\;
  \frac{p_{S_T}((S_0)^2/S_T)}{p_{S_T}(S_T)}\Big]\,.
\end{displaymath}
Apart from trivial cases, the density $p_{S_T}$ of $S_T$ depends
on $T$. In view of applications to semi-static hedging described
in \cite{car:lee08} it is beneficial if the correcting expression
\begin{displaymath}
  q_{S_t}(x)=\frac{p_{S_t}((S_0)^2/x)}{p_{S_t}(x)}
\end{displaymath}
at any time $t\in[0,T]$ depends only on $x$ and $S_0$ but not on
$t$. This is the case if $\eta$ is self-dual with no carrying
costs (then $q_{S_t}(x)=(x/S_0)^3$, $x>0$, by
Theorem~\ref{th:gen-sym}), or \emph{quasi-self-dual} with
parameters $a=e^\lambda$ and some $\alpha\neq 0$, being the case
if and only if $q_{S_t}(x)=(x/S_0)^{2+\alpha}$, $x>0$. In the
latter case~(\ref{eq:qsd}) turns into
\begin{displaymath}
  \E f(F\eta) =\E\Big[f\big(\frac{F}{a^2\eta}\big)a^\alpha\eta^\alpha\Big]\,.
\end{displaymath}
By letting $f(x)=(x-k)_+^\alpha$ and noticing that
$\E(a\eta)^\alpha=1$ in the quasi-self-dual case, this implies the
following property
\begin{displaymath}
  \E(F\eta-k)_+^\alpha=a^{-\alpha}\E (F-ka^2\eta)_+^\alpha
  =\E\eta^\alpha \; \E(F-k\eta(\E\eta^\alpha)^{-2/\alpha})_+\,,
\end{displaymath}
which can be termed as the \emph{power put-call symmetry}.

\section{Barrier options and semi-static hedging}
\label{sec:ex-apl-self-dual}

\subsection{Time-dependent framework}
\label{sec:time-depend-fram}

Consider a finite horizon model with the asset prices given by
\begin{displaymath}
  S_t=S_0\circ e^{t\lambda}\circ \eta_t=S_0\circ e^{t\lambda+\xi_t}
  =(S_{01}e^{t\lambda_1+\xi_{t1}},\dots,S_{0n}e^{t\lambda_n+\xi_{tn}})
  \,,\quad t\in[0,T]\,,
\end{displaymath}
where $\lambda\in\R^n$ represent deterministic carrying costs and
all components of $\eta_t=e^{\xi_t}$ are martingales with $\xi_t$,
$t\in[0,T]$, being a L\'evy process. Fix $i\in\{1,\dots,n\}$ and
assume that $\eta_t\in\QSD_i(t\lambda,\alpha)$ for every
$t\in[0,T]$. This condition is satisfied (with $\alpha=1$ and
$\lambda=0$) for all exponentially self-dual L\'evy models with no
carrying costs analysed in Section~\ref{sec:expon-self-dual-1} and
for quasi-self-dual L\'evy models from
Section~\ref{sec:quasi-self-dual-1} for non-vanishing $\lambda$, see
Corollary~\ref{cor:all-T}.

Let $\tau$ be a stopping time with values in $[0,T]$ and let
$\salg_\tau$ be the corresponding stopping $\sigma$-algebra. Since
$\xi_t$ is a L\'evy process, $(\xi_\tau,\xi_T)$ and
$(\xi_\tau,\xi_\tau+\xi'_{T-\tau})$ share the same distribution,
where $\xi'_t$, $t\in[0,T]$, is an independent copy of the process
$\xi_t$, $t\in[0,T]$. Hence, $(S_\tau,S_T)$ and $(S_\tau,S_\tau\circ
e^{\lambda(T-\tau)+\xi'_{T-\tau}})$ also coincide in distribution.
Then
\begin{align*}
  \E[f(S_T)|\salg_\tau]=\E[f(S_T)|S_\tau]
  &=\E[f(S_\tau\circ e^{\lambda(T-\tau)+\xi'_{T-\tau}})|S_\tau]\\
  &=\E[f(S_\tau\circ e^{\lambda(T-\tau)+\xi'_{T-\tau}})|\salg_\tau]\,,
\end{align*}
where $f$ is any integrable payoff function.  The quasi-self-duality
of $\eta'_{T-\tau}=e^{\xi'_{T-\tau}}$ with respect to the $i$th
numeraire adjusted for conditional expectations (see
Remark~\ref{rem:tau-T-mult}) yields that
\begin{displaymath}
  \E[f(S_T)|\salg_\tau]
  =\E[f(S_\tau\circ e^{K_i(\lambda (T-\tau)+\xi'_{T-\tau})})
    e^{\alpha(\lambda_i(T-\tau)+\xi'_{(T-\tau)i})} |\salg_\tau]\,,
\end{displaymath}
whence
\begin{multline}
  \label{eq:tau-sym}
  \E[f(S_T)|\salg_\tau]
  \\=\E\Big[f\Big(\frac{S_{T1}S_{\tau i}}{S_{Ti}},\dots,
  \frac{S_{T(i-1)}S_{\tau i}}{S_{Ti}},
  \frac{S_{\tau i}^2}{S_{Ti}},
  \frac{S_{T(i+1)}S_{\tau i}}{S_{Ti}},\dots,
  \frac{S_{Tn}S_{\tau i}}{S_{Ti}}\Big)
  \Big(\frac{S_{Ti}}{S_{\tau i}}\Big)^\alpha|\salg_\tau\Big]\,,
\end{multline}
cf.\ Remark~\ref{rem:sd-forward}. If $S_{\tau i}=H$ almost surely
for a constant $H$, then
\begin{multline}
  \label{eq:camv}
  \E[f(S_T)|\salg_\tau]\\=\E\Big[f\Big(\frac{S_{T1}H}{S_{Ti}},\dots,
  \frac{S_{T(i-1)}H}{S_{Ti}},\frac{H^2}{S_{Ti}},
  \frac{S_{T(i+1)}H}{S_{Ti}},\dots,
  \frac{S_{Tn}H}{S_{Ti}}\Big)\Big(\frac{S_{Ti}}{H}\Big)^\alpha|\salg_\tau\Big].
\end{multline}
Identity~(\ref{eq:camv}) for $\alpha=1$, $\lambda=0$ yields the
self-dual case, and in the univariate case $n=1$ corresponds
to~\cite[Eq.~(5.3)]{car:lee08}. Classical examples with trivial
carrying costs (i.e.\ $\lambda=0$) are options on futures or options
on shares with dividend-yield being equal to the risk-free interest
rate. For the univariate quasi-self-dual case,
see~\cite[Cor.~5.10]{car:lee08}.

\begin{remark}
  \label{rem:t-c}
  Instead of the self-duality property, it is possible to
  impose~(\ref{eq:tau-sym}) for stopping times $\tau\in[0,T]$ that
  might appear in relation to hedging of particular barrier options.
  This observation leads to extensions for independently time-changed
  multidimensional L\'evy processes by means of conditioning arguments
  described in~\cite[Th.~4.2,~5.4]{car:lee08}.  Further
  models without jumps can be obtained on the basis of the multivariate
  Black-Scholes model with characteristics described in
  Example~\ref{ex:log-n-mult} by applying independent common
  stochastic clocks being continuous with respect to
  calendar time.
\end{remark}

\subsection{Multivariate hedging with a single univariate barrier}
\label{sec:general-hedging-with}

Assume a risk-neutral setting for a price process $S_t$,
$t\in[0,T]$, and fix a barrier level at $H>0$ such that $S_{0i}\neq
H$ with given $i\in\{1,\dots,n\}$. For simplicity of notation,
define function $\hat\kappa_i:\EE^n\mapsto\EE^n$ acting as
\begin{displaymath}
  \hat\kappa_i(S_T,H)=\frac{H}{S_{Ti}}\big(S_{T1},\dots,
  S_{T(i-1)},H,S_{T(i+1)},\dots,S_{Tn}\big)\,.
\end{displaymath}
Define $\Xi_{ti}$ to be the (closed) line segment with end-points
$S_{0i}$ and $S_{ti}$ and let
\begin{displaymath}
  \tau_H=\inf\{t\geq0:\; H\in\Xi_{ti}\}\quad \text{and}\quad
  \chi=\one_{\tau_H\leq T}\,,
\end{displaymath}
cf.~\cite[Sec.~5.2]{car:lee08} who used a bit different way to
handle the two cases when the initial price is respectively lower
and higher than the barrier. Furthermore, assume that the asset
price dynamics satisfy~(\ref{eq:tau-sym}) for the stopping time
$\tau=\tau_H$ and that $S_{\tau_Hi}=H$\ a.s.\ on the event that
$\{\tau_H\leq T\}$, what is guaranteed, e.g.\ by the sample path
continuity of the $i$th component of $S_t$, $t\geq0$. In case of
discontinuous L\'evy processes, the symmetry
condition~(\ref{eq:cond-2}) on the L\'evy measure implies the
presence of jumps of both signs, so it is much more difficult to
ensure that $S_{\tau_H}=H$ a.s.

Take any integrable payoff function $f$ and consider an option with
payoff $\chi f(S_T)$, i.e.\ the knock-in option with barrier $H$ for
the $i$th asset.  In order to replicate this option using only
options that depend on the terminal value $S_T$ consider a European
claim on
\begin{equation}
  \label{eq:hedge-claim}
  f(S_T)\one_{H\in \Xi_{Ti}}
  +\Big(\frac{S_{Ti}}{H}\Big)^\alpha f(\hat\kappa_i(S_T,H))
  (\one_{H\in\Xi_{Ti}}-\one_{S_{Ti}=H})\,.
\end{equation}
Here one has to bear in mind that this is only practicable provided
that the considered claims are liquid or can be replicated by liquid
instruments. However, there is a fast growing literature about sub-
and super-replication of multiasset instruments, see
e.g.~\cite{lau:wan08} and the literature cited therein.

On the event that $\{\tau_H>T\}$, the claim
in~(\ref{eq:hedge-claim}) expires worthless as desired. If the
barrier knocks in, we can exchange (\ref{eq:hedge-claim}) for a
claim on $f(S_T)$ at zero costs. To confirm this, define
$\hat\Xi_{ti}$ to be the (closed) line segment with end-points
$S_{0i}$ and $H^{2}S_{ti}^{-1}$. Note that $H\notin \hat\Xi_{ti}$ if
and only if $H\in\Xi_{ti}\setminus\{S_{ti}\}$. Hence, on the event
that $\{\tau_H\leq T\}$, by~(\ref{eq:camv}), we have
\begin{align*}
  \E[f(S_T)|\salg_\tau]&=\E[f(S_T)\one_{H\in \Xi_{Ti}}|\salg_\tau]
  +\E[f(S_T)\one_{H\notin \Xi_{Ti}}|\salg_\tau]\\
   &=\E[f(S_T)\one_{H\in \Xi_{Ti}}|\salg_\tau]
  +\E\Big[\big(\frac{S_{Ti}}{H}\big)^\alpha f(\hat\kappa_i(S_T,H))
  \one_{H\notin\hat\Xi_{Ti}}|\salg_\tau\Big]\\
   &=\E[f(S_T)\one_{H\in \Xi_{Ti}}|\salg_\tau]\\
   &\qquad +\E\Big[\big(\frac{S_{Ti}}{H}\big)^\alpha f(\hat\kappa_i(S_T,H))
  (\one_{H\in \Xi_{Ti}}-\one_{S_{Ti}=H})|\salg_\tau\Big]\,.\\
\end{align*}

For simplicity we assume from now on that $S_{Ti}$ has a non-atomic
distribution, so that~(\ref{eq:hedge-claim}) becomes
\begin{equation}
  \label{eq:hedge-claim-abs-cont}
  \Big(f(S_T)+\big(\frac{S_{Ti}}{H}\big)^\alpha
  f(\hat\kappa_i(S_T,H))\Big)\one_{H\in \Xi_{Ti}}\,.
\end{equation}

Consider general basket call $f(S_T)=(\sum_{j=1}^n u_jS_{Tj}-k)_+$.
By (\ref{eq:hedge-claim-abs-cont}) the hedge for the knock-in basket
call with payoff function $\chi f(S_T)$ is given by the derivative
with payoff function
\begin{displaymath}
  \Big\{\Big(\sum_{j=1}^n u_jS_{Tj}-k\Big)_+
  +\Big(\frac{S_{Ti}}{H}\Big)^{\alpha-1}\Big(u_iH-\big(\frac{k}{H}S_{Ti}
  -\sum_{j=1,\; j\neq i}^n u_jS_{Tj}\big)\Big)_+\Big\}\one_{H\in\Xi_{Ti}}\,,
\end{displaymath}
which depends only on $S_T$.

If (\ref{eq:tau-sym}) holds with $\alpha=1$, this hedge becomes
\begin{equation}
  \label{eq:bask-hedges}
  \Big\{\Big(\sum_{j=1}^n u_jS_{Tj}-k\Big)_+
  +\Big(u_iH-\big(\frac{k}{H}S_{Ti}
  -\sum_{j=1,\; j\neq i}^n
  u_jS_{Tj}\big)\Big)_+\Big\}\one_{H\in\Xi_{Ti}}\,,
\end{equation}
being the sum of a basket call and a spread put with knocking
condition depending only on the $i$th component $S_{Ti}$ at
maturity.

In some cases the knocking condition at maturity can be incorporated
into the payoff function. For this, note that we can write all
integrable payoff functions in the form
\begin{displaymath}
  f(S_T)=f_0(S_T)\one_{S_T\in\Theta}\,,
  \quad\text{ where }\quad\Theta=\{x:f(x)\neq 0\}\,,
\end{displaymath}
with $\Theta$ possibly being $\EE^n$. For example, the basket call
$(\sum_{j=1}^n u_jS_{Tj}-k)_+$ can be written as $(\sum_{j=1}^n
u_jS_{Tj}-k)\one_{\sum_{j=1}^n u_jS_{Tj}>k}$. Of course if $H\notin
\Xi_{Ti}$ would imply that $\hat\kappa_i(S_T,H)\notin\Theta$ and
$S_T\notin\Theta$ at the same time, then it is possible to omit
$\one_{H\in \Xi_{Ti}}$ in~(\ref{eq:hedge-claim-abs-cont}), but this
is not the case for standard payoff functions. If $H\in\Xi_{Ti}$
implies that $S_T\notin\Theta$ (resp.\
$\hat\kappa_i(S_T,H)\notin\Theta$) then the first (second) summand
in~(\ref{eq:hedge-claim-abs-cont}) is always zero. If furthermore
$H\notin\Xi_{Ti}$ implies that $\hat\kappa_i(S_T,H)\notin\Theta$
($S_T\notin\Theta$) then we can omit the first (second) summand
in~(\ref{eq:hedge-claim-abs-cont}) and hedge with the second (first)
summand without the knocking condition $\one_{H\in\Xi_{Ti}}$, i.e.\
in~(\ref{eq:bask-hedges}) we can hedge with a conventional basket
option.

\begin{example}
  \label{ex:sp-opt}
  Consider a bivariate price process $(S_{t1},S_{t2})$ in a
  risk-neutral setting satisfying~(\ref{eq:tau-sym}) with $\alpha=1$
  and $i=1$ for the stopping time $\tau=\tau_H=\inf\{t:S_{t1}\leq H\}$
  with barrier $H$ such that $0<H<S_{01}$. First assume again that
  $S_{t1}$ can not jump over $H$. For the spread option
  \begin{displaymath}
    f(S_{T1},S_{T2})=(aS_{T1}-bS_{T2}-k)_+\,,\quad a,b>0\,,
  \end{displaymath}
  assume additionally that $aH\leq k$ and define
  $\chi=\one_{\tau_H\leq T}$. By using the hedging strategy described
  in~(\ref{eq:hedge-claim-abs-cont}) and $aH\leq k$ we obtain a henge for
  $\chi f(S_{T1},S_{T2})$ by
  \begin{multline*}
    (aS_{T1}-bS_{T2}-k)_+\one_{H\in [S_{T1},S_{01}]}
    +\Big(aH-\frac{k}{H}S_{T1}-bS_{T2}\Big)_+\one_{H\in [S_{T1},S_{01}]}\\
    =\Big(aH-\frac{k}{H}S_{T1}-bS_{T2}\Big)_+\one_{H\in [S_{T1},S_{01}]}
    =\Big(aH-\frac{k}{H}S_{T1}-bS_{T2}\Big)_+\,,
  \end{multline*}
  i.e.\ it is possible to hedge with a basket put. Therefore, the
  related knock-out option can be hedged with a long position in the
  spread call with payoff function
  $f(S_{T1},S_{T2})=(aS_{T1}-bS_{T2}-k)_+$ and a short position in the
  above hedge. Note that we only assumed that $b>0$ so that the
  knock-in level can but need not be deep out-of-the-money. If
  $S_{t1}$ can jump over the barrier $H$ we get a
  super-replication in case of the knock-in option and a more
  problematic sub-replication in case of the knock-out option.

  Assuming~(\ref{eq:tau-sym}) for the stopping time $\tau_H$ with
  $\alpha\neq 1$, where $S_{t1}$ does not jump over $H$, the hedge
  for the knock-in option has to be modified as
  \begin{displaymath}
    \Big(\frac{S_{T1}}{H}\Big)^{\alpha-1}
    \Big(aH-\frac{k}{H}S_{T1}-bS_{T2}\Big)_+\,,
  \end{displaymath}
  while the modification for the knock-out is now obvious.
  This example can easily be extended in a higher dimensional setting
  as long as all risky assets without knocking barrier enter the
  payoff function with a minus sign.
\end{example}

\begin{example}
  \label{ex:opt-peculiar}
  Consider a bivariate price process in a risk-neutral setting. Assume
  that $S_{01}>k>0$ and that~(\ref{eq:tau-sym}) holds for
  $\tau=\tau_k=\inf\{t:S_{t1}\leq k\}$, $\alpha=1$, and $i=1$,
  while $S_{t1}$ can not jump over $k$. Introduce
  (possibly negative) payoff function
  \begin{displaymath}
    g(S_{T1},S_{T2})=(S_{T1}-k)+(S_{T2}\wedge(S_{T1}-k))\,,
  \end{displaymath}
  where $a\wedge b=\min(a,b)$.

  By~(\ref{eq:hedge-claim-abs-cont}) with $\alpha =1$ we obtain a
  hedge for $\one_{\tau_k\leq T} g(S_{T1},S_{T2})$ as
  \begin{displaymath}
    \Big\{\big((S_{T1}-k)+(S_{T2}\wedge(S_{T1}-k))\big)
    +\big((k-S_{T1})+(S_{T2}\wedge(k-S_{T1}))\big)\Big\}
    \one_{k\in[S_{T1},S_{01}]}\,.
  \end{displaymath}
  Here we can get rid of the indicator function by noticing that the
  above payoff function can be written as
\begin{displaymath}
  \big((S_{T1}-k)+(S_{T2}\wedge(S_{T1}-k))\big)-(S_{T1}+S_{T2}-k)_+
  +(k+S_{T2}-S_{T1})_+\,.
\end{displaymath}
Furthermore, for the related knock-out option we get a hedge given
by a long position in the basket call with payoff function
$(S_{T1}+S_{T2}-k)_+$ and a short position in the put spread with
payoff function $(k-(S_{T1}-S_{T2}))_+$.
\end{example}

\subsection{Examples of hedging in jointly self-dual cases}

In this section we create hedges for more complex instruments and
bivariate models satisfying~(\ref{eq:tau-sym}) for two different
numeraires, e.g.\ for jointly self-dual exponential L\'evy models
with generating triplets satisfying the conditions in
Remark~\ref{rem:scd}.

\begin{example}[Options with knocking-conditions depending on two assets]
  \label{eg:knock-on-two}
  Assume a risk-neutral setting for a price process
  $S_t=(S_{t1},S_{t2})$, $t\in[0,T]$, where (\ref{eq:tau-sym}) holds
  for both assets with $\alpha=1$ and the subsequently defined
  stopping times.  Furthermore, let $k_x,k_y>0$ be constants such that
  $k_x<S_{01},S_{02}<k_y$ and define the stopping times
  $\tau_{ix}=\inf\{t>0:S_{ti}\leq k_x\}$,
  $\tau_{iy}=\inf\{t>0:S_{ti}\geq k_y\}$ and the corresponding
  stopping $\sigma$-algebras $\salg_{\tau_{ix}}$, $\salg_{\tau_{iy}}$,
  $i=1,2$, as well as the stopping time
  $\tau=\tau_{1x}\wedge\tau_{2x}$ with the corresponding stopping
  $\sigma$-algebra $\salg_\tau$. Assume also that the price processes
  can not jump over the barriers $k_x$ and $k_y$ respectively.

  Consider the claims
  \begin{align*}
    X&=(S_{T1}-S_{T2}-k_x)_+\one_{\tau=\tau_{1x}\leq T}
       +(S_{T2}-S_{T1}-k_x)_+\one_{\tau_{1x}\neq\tau=\tau_{2x}\leq T}\,,\\
    Y&=(k_y-S_{T1}-S_{T2})_+\big(\one_{\tau_{1y}\leq T<\tau_{2y}}
       -\one_{\tau_{2y}\leq T<\tau_{1y}}\big)\,,
  \end{align*}
  i.e.\ $X$ is a knock-in spread option on the difference between the
  share which first hits $k_x$ and the other one only being knocked-in
  if at least one share hits $k_x$ before $T$. At maturity, $Y$ is
  a long position in a basket put if and only if the first but not the
  second asset price hits the price level $k_y$ and a short position in
  the same basket put if and only the second but not the first asset hits
  $k_y$ before $T$.

  The claim $X$ can be hedged with a long position in the European
  basket put with payoff $(k_x-S_{T1}-S_{T2})_+$.
  The claim $Y$ can be hedged by entering a long position in the spread option
  with payoff $(S_{T1}-S_{T2}-k_y)_+$ along with a short position in the spread
  option with payoff $(S_{T2}-S_{T1}-k_y)_+$. To see that, apply
  identity~(\ref{eq:camv}) for $\alpha=1$, so that
  \begin{align}
    \label{eq:hedge-long-joint}
    \E[(k_z-S_{T1}-S_{T2})_+|\salg_{\tau_{1z}}]
      &=\E[(S_{T1}-S_{T2}-k_z)_+|\salg_{\tau_{1z}}]\,,\\
    \label{eq:hedge-short-joint}
    \E[(k_z-S_{T1}-S_{T2})_+|\salg_{\tau_{2z}}]
      &=\E[(S_{T2}-S_{T1}-k_z)_+|\salg_{\tau_{2z}}]\,,
  \end{align}
  $z=x,y$, while the value of the European spread option with payoff
  $(S_{T1}-S_{T2}-k)_+$ (resp.\ $(S_{T2}-S_{T1}-k)_+$) remains
  unchanged by applying~(\ref{eq:camv}) at $\tau_2$ (resp.\
  $\tau_1$).

  As far as $X$ is concerned we have that in case where
  $\{\tau >T\}$ neither $X$ is knocked in nor the basket-put in
  the hedge portfolio is in the money, since $S_{T1},S_{T2}>k_x$.
  If $\{\tau=\tau_{1x}\leq T\}$, by~(\ref{eq:hedge-long-joint}) we can
  exchange the long position in the hedge portfolio for the needed
  spread, if $\tau_{1x}\neq\tau=\tau_{2x}\leq T$ the same is true
  due to~(\ref{eq:hedge-short-joint}).

  As far as $Y$ is concerned we have that if $\{\tau_{iy}>T\}$,
  $i=1,2$, both instruments in the hedging portfolio are out of the
  money since $S_{T1},S_{T2}\leq k_y$, i.e.\ have payoff zero as $Y$.
  On the event that we first have $\{\tau_{1y}\leq T\}$ we can change
  the long position in the spread option with payoff
  $(S_{T1}-S_{T2}-k_y)_+$ in a long position in the basket put with
  payoff $(k_y-S_{T1}-S_{T2})_+$ while letting the short position
  unchanged. Provided that additionally
  $\{\tau_{2y}\in[\tau_{1y},T]\}$, by~(\ref{eq:hedge-short-joint}), we
  can also exchange the short position for the same basket put, so
  that we can close our positions as required, otherwise, i.e.\
  $\{\tau_{2y}>T\}$, the long position in the hedge portfolio yields a
  potentially needed payoff of $Y$ while the short position still
  matures worthless. On the event that we first have $\{\tau_{2y}\leq
  T\}$, by~(\ref{eq:hedge-short-joint}), we can change the short
  position in the needed basket put while letting the long position
  unchanged. If furthermore $\{\tau_{1y}\in[\tau_{2y},T]\}$ we can
  again close our position due to~(\ref{eq:hedge-long-joint}).
  Otherwise, unlike the long position, the short position in the hedge
  portfolio may be in the money at maturity but in that case $Y$ at
  maturity would be a long position for the hedger with the same
  payoff.

  Assuming that~(\ref{eq:tau-sym}) holds for $\alpha_1$ resp.\
  $\alpha_2$ with respect to the corresponding components (where other
  assumptions remain unchanged), we can
  apply~(\ref{eq:camv}) in the same way to see that the hedge should
  theoretically be modified to a long position in the European
  derivative with payoff
  $(S_{T1}-S_{T2}-k_y)_+(k_y^{-1}S_{T1})^{\alpha_1-1}$ along with a
  short position in the European derivative with payoff
  $(S_{T2}-S_{T1}-k_y)_+(k_y^{-1}S_{T2})^{\alpha_2-1}$.
\end{example}

For creating semi-static hedges of barrier spread-options with
certain knocking conditions, e.g.\ claims of the form
\begin{displaymath}
  Z=(S_{T1}-S_{T2}-k)_+\one_{S_{t1}>S_{t2},\forall
    t\in[0,T]}\,, \quad k>0\,,
\end{displaymath}
in equal carrying cost cases the full strength of the joint
self-duality is not needed. It suffices to assume exchangeability
being implied by the joint self-duality, see
Corollary~\ref{cor:comp} and~\cite{mol:sch09} for details including
model characterisations, weakening of the exchangeability assumption
and hedges for several related derivatives.

\subsection{Semi-static super-hedges of basket options}
\label{sec:semi-static-super}

The following super-hedges may be quite expensive for replication
purpose. Thus, they seem to be more useful if one would like to
speculate with a basket option and get some additional money by
writing a different knock-in basket option, where the maximum loss
should be limited to the initially invested capital.

In the sequel we work in the same setting as in
Section~\ref{sec:general-hedging-with} with the additional
assumptions that $\alpha=1$ and $S_{0i}>H>0$. Define again the
stopping time $\tau_H=\inf\{t:S_{ti}\leq H\}$ and let
$\salg_{\tau_H}$ be the corresponding $\sigma$-algebra. Consider the
knock-in basket option with the following payoff function
\begin{displaymath}
  \chi\Big(\sum_{j=1}^n u_jS_{Tj}-k\Big)_+\,,\quad
  k,u_i>0,\ u_j\in\R\ \text{ for }\  j=1,\dots,n\,,\ j\neq i\,,
\end{displaymath}
where $\chi=\one_{\tau_H\leq T}$.

By~(\ref{eq:camv}) for $\alpha=1$ we have
\begin{displaymath}
  \E\Big[\big(\sum_{j=1}^nu_jS_{Tj}-k\big)_+|\salg_{\tau_H}\Big]
   =\E\Big[\big(\sum_{j=1,\; j\neq i}^n u_jS_{Tj}-\frac{k}{H}S_{Ti}
   +u_iH\big)_+|\salg_{\tau_H}\Big]\,.
\end{displaymath}
Hence, the maximum loss of buying the basket option in the
right-hand side of the above equation and short-selling the initial
knock-in basket call does not exceed the initial costs of this
strategy. Note that in this setting, jumps over the barrier $H$
would add a further aspect of super-hedging.

If~(\ref{eq:tau-sym}) holds for $\alpha\neq 1$ where $S_{ti}$ can
not jump over $H$, use again~(\ref{eq:camv}). The long position in
the described strategy would then be determined by the payoff
function
\begin{displaymath}
  \Big(\frac{S_{Ti}}{H}\Big)^{\alpha-1}
   \big(\sum_{j=1,\; j\neq i}^n u_jS_{Tj}-\frac{k}{H}S_{Ti}
    +u_iH\big)_+\,,
\end{displaymath}
while the short position remains unchanged.

\section*{Acknowledgements}
\label{sec:acknowledgements}

The authors are grateful to Daniel Neuenschwander and Pierre Patie
for helpful discussions, to Rolf Burgermeister, Peter Carr, and
Markus Liechti for useful hints, and to the referees for a number
of stimulating comments. This work was supported by the Swiss
National Science Foundation Grant Nr. 200021-117606.

\newcommand{\noopsort}[1]{} \newcommand{\printfirst}[2]{#1}
  \newcommand{\singleletter}[1]{#1} \newcommand{\switchargs}[2]{#2#1}

\end{document}